\documentclass[11pt, reqno]{amsart}
\usepackage{amsmath,amssymb}
\usepackage{graphicx,epsfig,epstopdf}
\usepackage{natbib}
\long\def\comment#1{}
\newtheorem{theorem}{Theorem}[section]

\newtheorem{lemma}{Lemma}[section]
\newtheorem{proposition}{Proposition}[section]

\theoremstyle{definition}

\newtheorem{remark}{Comment}[section]
\newtheorem{example}{Example}[section]

\DeclareMathOperator{\var}{var}

\newcommand{\be}{\begin{eqnarray}}
\newcommand{\ee}{\end{eqnarray}}

\newcommand{\ba}{\begin{array}}
\newcommand{\ea}{\end{array}}
\newcommand{\bs}{\begin{align}\begin{split}\nonumber}
\newcommand{\bsnumber}{\begin{align}\begin{split}}
\newcommand{\es}{\end{split}\end{align}}

\newcommand{\n}{n}

\renewcommand{\(}{\left(}
\renewcommand{\)}{\right)}
\renewcommand{\[}{\left[}
\renewcommand{\]}{\right]}

\newcommand{\G}{\mathcal{G}}

\newcommand{\X}{\mathcal{X}}
\newcommand{\I}{\mathcal{I}}
\newcommand{\xik}{\xi_k}

\newcommand{\Gn}{\mathbb{G}_n}

\newcommand{\En}{\mathbb{E}_n}

\def\RR{ {\Bbb{R}}}

\begin{document}

\title[Least Squares Series: Pointwise and Uniform Results]{Some New Asymptotic Theory for Least Squares Series: Pointwise and Uniform Results}

\author[Belloni]{Alexandre Belloni}\address[Alexandre Belloni]{Fuqua School of Business, Duke University, United States}
\author[Chernozhukov]{Victor Chernozhukov}\address[Victor Chernozhukov]{Department of Economics, MIT, United States}
\author[Chetverikov]{Denis Chetverikov}\address[Denis Chetverikov]{Department of Economics, UCLA, United States}\email{chetverikov@econ.ucla.edu }
\author[Kato]{Kengo Kato}\address[Kengo Kato]{Graduate School of Economics, The University of Tokyo, Japan}

\date{ First version:  May 2006,  This version is of  \today.   Submitted to ArXiv and for publication: December 3, 2012.
 }


\maketitle

\begin{abstract}
 In econometric applications it is common that the exact form of a conditional expectation is unknown and having  flexible functional forms can lead to improvements over a pre-specified functional form, especially if they nest some successful parametric economically-motivated forms.   Series method offers exactly that by approximating the unknown function based on $k$ basis functions, where $k$ is allowed to grow with the sample size $n$ to balance the trade off between variance and bias.
 In this work we consider series estimators for the conditional mean in light of four new ingredients: (i) sharp LLNs for matrices derived from the non-commutative Khinchin inequalities, (ii) bounds on the Lebesgue factor that controls the ratio between the $L^{\infty}$ and $L^{2}$-norms of approximation errors,  (iii) maximal inequalities for processes whose entropy integrals diverge at some rate, and (iv) strong approximations to series-type processes.

These technical tools allow us to contribute to the series literature, specifically the seminal work of \cite{Newey1997}, as follows. First, we weaken considerably the condition on the number $k$ of approximating functions used in series estimation
from the typical  $k^2/n \to 0$ to $k/n \to 0$, up to
log factors, which was available only for spline  series before. Second, under the same weak conditions we derive $L^{2}$ rates and pointwise central limit theorems results when the approximation error vanishes. Under an incorrectly specified model, i.e. when the approximation error does not vanish, analogous results are also shown. Third, under stronger conditions we derive uniform rates and functional central limit theorems that hold if the approximation error vanishes or not. That is, we derive the strong approximation for the entire estimate of the nonparametric function.

Finally and most importantly, from a point of view of practice, we derive uniform rates, Gaussian approximations, and uniform confidence bands for a wide collection of linear functionals of the conditional expectation function, for example, the function itself,  the partial derivative function, the conditional average partial derivative function, and other similar quantities.  All of these results are new.
\end{abstract}

\section{Introduction}

Series estimators have been playing a central role in various fields. In econometric applications it is common that the exact form of a conditional expectation is unknown and having a flexible functional form can lead to improvements over a pre-specified functional form, especially if it nests some successful parametric economic models.   Series estimation offers exactly that by approximating the unknown function based on $k$ basis functions, where $k$ is allowed to grow with the sample size $n$ to balance the trade off between variance and bias. Moreover, the series modelling allows for convenient nesting of some theory-based models, by simply using corresponding terms as the first $k_0 \leq k$
basis functions. For instance, our series could contain linear and quadratic functions  to nest the canonical Mincer equations in the context of wage equation modelling or the canonical translog demand and production functions in the context of demand and supply modelling.

Several asymptotic properties of series estimators have been investigated in the literature. The focus has been on convergence rates and asymptotic normality results \citep[see][and the references therein]{vandeGeer1990,Andrews1991,EastwoodGallant1991,GallantSouza1991,Newey1997,vandeGeer2002,Huang2003b,Chen2006,CF2013}.

This work revisits the topic by making use of new critical ingredients:
\begin{itemize}
\item[1.] The sharp LLNs for matrices derived from the non-commutative Khinchin inequalities.
\item[2.] The sharp bounds on the Lebesgue factor that controls the ratio between the $L^{\infty}$ and $L^{2}$-norms
of the least squares approximation of functions (which is bounded or grows like a $\log k$ in many cases).
\item[3.] Sharp maximal inequalities for processes whose entropy integrals diverge at some rate.
\item[4.] Strong approximations to empirical processes of series types.
\end{itemize}
To the best of our knowledge, our results are the first applications of the first ingredient to statistical estimation problems. After the use in this work, some recent working papers are also using related matrix inequalities and extending some results in different directions, e.g. Chen and Christensen (2013) allows $\beta$-mixing dependence, and Hansen (2014) handles unbounded regressors and also characterizes a trade-off between the number of finite moments and the allowable rate of expansion of the number of series terms.  Regarding the second ingredient, it has already been used by \cite{Huang2003} but for splines only.  All of these ingredients are critical for generating sharp results.

This approach allows us to contribute to the series literature in several directions. First, we weaken considerably the condition on the number $k$ of approximating functions used in series estimation
from the typical  $k^2/n \to 0$ \citep[see][]{Newey1997} to
$$
k/n \to 0 \text{ (up to logs)}
$$
for bounded or local bases which was previously available only for spline series \citep{Huang2003, Stone1994}, and recently established for local polynomial partition series \citep{CF2013}. An example of a bounded basis is Fourier series; examples of local bases are spline, wavelet, and local polynomial partition series. To be more specific, for such bases we require $k\log k/n\to 0$. Note that the last condition is similar to the condition on the bandwidth value required for local polynomial (kernel) regression estimators ($h^{-d}\log (1/h)/n\to 0$ where $h=1/k^{1/d}$ is the bandwidth value).
Second, under the same weak conditions we derive $L^{2}$ rates and pointwise central limit theorems results when the approximation error vanishes. Under a misspecified model, i.e. when the approximation error does not vanish, analogous results are also shown.
Third, under stronger conditions we derive uniform rates that hold if the approximation error vanishes or not. An important contribution here is that we show that the series estimator achieves the optimal uniform rate of convergence under quite general conditions. Previously, the same result was shown only for local polynomial partition series estimator \citep{CF2013}. In addition, we derive a functional central limit theorem. By the functional central limit theorem we mean here that the entire estimate of the nonparametric function is
uniformly close to a Gaussian process that can change with $n$.  That is, we derive the strong approximation for the entire estimate of the nonparametric function.

Perhaps the most important contribution of the paper is a set of completely new results that
 provide estimation and inference methods
for the \textit{entire} linear functionals $\theta(\cdot)$ of the conditional mean function $g:\X\to \RR$. Examples of linear functionals $\theta(\cdot)$ of interest include
\begin{itemize}
\item[1.]  the partial derivative function:  \quad $x \mapsto \theta(x) = \partial_j g(x)$;
\item[2.]  the average partial derivative:  \quad $\theta = \int \partial_j g(x) d\mu(x)$;
\item[3.]  the conditional average partial derivative: \ $x^s \mapsto \theta(x^s) = \int \partial_j g(x) d\mu(x|x^s)$.
\end{itemize}
where $\partial_j g(x)$ denotes the partial derivative of $g(x)$ with respect to $j$th component of $x$, $x^s$ is a subvector of $x$, and the measure $\mu$ entering the definitions above is taken as known; the result can be extended to include estimated measures. We derive uniform (in $x$) rates of convergence, large sample
distributional approximations, and inference methods for the functions above based on the Gaussian approximation. To the best of our knowledge all these results are new, especially the distributional and inferential results.  For example, using these results we can now perform inference on the entire partial derivative function. The only other reference that provides analogous results  but for quantile series estimator is \cite{BCF2011}.  Before doing uniform analysis,  we also update the pointwise results of \cite{Newey1997} to weaker, more general conditions.

\textbf{Notation.} In what follows, all parameter values are indexed by the sample size $n$, but we omit the index whenever this does not cause
confusion.  We use the notation $(a)_+ = \max\{a,0\}$, $a \vee b = \max\{ a, b\}$ and $a \wedge b = \min\{ a , b \}$. The $\ell_2$-norm of a vector $v$ is denoted by $\|v\|$, while for a matrix $Q$ the operator norm is denoted by $\|Q\|$. We also use standard notation in the empirical process literature,
$$
\En[f] = \En[f(w_i)] = \frac{1}{n}\sum_{i=1}^n f(w_i)\text{ and }\mathbb{G}_n[f]=\mathbb{G}_n[f(w_i)]= \frac{1}{\sqrt{n}}\sum_{i=1}^n (f(w_i) - E[f(w_i)])
$$
and we use the notation $a \lesssim b$ to denote $a \leqslant c b$ for some constant $c>0$ that does not depend on $n$; and
$a\lesssim_P b$ to denote $a=O_P(b)$. Moreover, for two random variables $X, Y$ we say that $X=_dY$ if they have the same probability distribution. Finally, $S^{k-1}$ denotes the space of vectors $\alpha$ in $\mathbb{R}^k$ with unit Euclidean norm: $\|\alpha\|=1$.

\section{Set-Up}\label{Sec:Setup}

Throughout the paper, we consider a sequence of models, indexed by the sample size $n$,
\begin{equation}\label{Setup}
y_i = g(x_i) + \epsilon_i,  \ \ E[\epsilon_i|x_i]=0, \ \ x_i \in \X \subseteq \Bbb{R}^d, \ \ i=1,\ldots,n,
\end{equation}
where $y_i$ is a response variable, $x_i$ a vector of covariates (basic regressors), $\epsilon_i$ noise, and $x\mapsto g(x) = E[y_i|x_i =x]$ a regression (conditional mean) function; that is, we consider a triangular array of models with $y_i=y_{i,n}$, $x_i=x_{i,n}$, $\epsilon_i=\epsilon_{i,n}$, and $g=g_n$. We assume that $g\in\mathcal{G}$ where $\mathcal{G}$ is some class of functions. Since we consider a sequence of models indexed by $n$, we allow the function class $\mathcal{G}=\mathcal{G}_n$, where the regression function $g$ belongs to, to depend on $n$ as well. In addition, we allow $\X=\X_n$ to depend on $n$ but we assume for the sake of simplicity that the diameter of $\X$ is bounded from above uniformly over $n$ (dropping the uniform boundedness condition is possible at the expense of more technicalities; for example, without uniform boundedness condition, we would have an additional term $\log \text{diam}(\X)$ in (\ref{eq: lin2U}) and (\ref{eq: lin4U}) of Lemma \ref{lemma: uniform linearization}). 
We denote $\sigma_i^2=E[\epsilon_i^2|x_i]$, $\bar{\sigma}^2:=\sup_{x\in\mathcal{X}}E[\epsilon_i^2|x_i=x]$, and $\underline{\sigma}^2:=\inf_{x\in\mathcal{X}}E[\epsilon_i^2|x_i=x]$.
For notational convenience, we omit indexing by $n$ where it does not lead to confusion.

\textbf{Condition A.1} (Sample) \textit{For each $n$, random vectors $(y_i, x_i')'$, $i=1,\ldots,n,$ are i.i.d. and satisfy (\ref{Setup}).}

We approximate the function $x\mapsto g(x)$ by linear forms $x \mapsto p(x)' b$, where
$$
x \mapsto p(x) := (p_{1}(x),\ldots,p_k(x))'
$$
is a vector of approximating functions that can change with $n$; in particular, $k$ may increase with $n$. We denote the regressors as
$$
p_i := p(x_i):= (p_{1}(x_i),\ldots,p_k(x_i))'.
$$
The next assumption imposes regularity conditions on the regressors.

\textbf{Condition A.2} (Eigenvalues) \textit{Uniformly over all $n$, eigenvalues of $Q:= E[p_i p_i']$ are bounded above and away from zero.}

Condition A.2 imposes the restriction that $p_1(x_i),\dots,p_k(x_i)$ are not too co-linear.  Given this assumption, it is without loss of generality to impose the following normalization:

\textbf{ Normalization.} \textit{To simplify notation, we normalize $Q = I$, but we shall treat $Q$ as unknown, that is
we deal with random design.}

The following proposition establishes a simple sufficient condition for A.2 based on orthonormal bases with respect to some measure.
\begin{proposition}[Stability of Bounds on Eigenvalues]\label{lemma:StabilityValues} Assume that $x_i\sim F$ where $F$ is a probability measure on $\X$, and that the regressors $p_1(x),\dots,p_k(x)$ are orthonormal on $(\X, \mu)$ for some measure $\mu$. Then
A.2 is satisfied if $dF/d\mu \text{ is bounded above and away from zero. }$
\end{proposition}

It is well known that the least squares parameter $\beta$ is defined by
$$
\beta := \arg \min_{b \in \Bbb{R}^k} E\left[(y_i - p_i'b)^2\right],
$$
which by (\ref{Setup}) also implies that $\beta=\beta_g$ where $\beta_g$ is defined by
\begin{equation}\label{eq: beta g def}
\beta_g:=\arg \min_{b \in \Bbb{R}^k} E\left[(g(x_i) - p_i'b)^2\right].
\end{equation}
We call $x \mapsto g(x)$ the target function and $x\mapsto g_k(x) = p(x)'\beta$
the surrogate function. In this setting, the surrogate function provides the best linear approximation to the target function.

For all $x\in\mathcal{X}$, let
\begin{equation}\label{eq: approximation error}
r(x):=r_g(x):=g(x)-p(x)'\beta_g
\end{equation}
denote the approximation error at the point $x$, and let
$$
r_i := r(x_i)= g(x_i) - p(x_i)'\beta_g
$$
denote the approximation error for the observation $i$.
Using this notation, we obtain a many regressors model
$$
y_i = p_i'\beta + u_i, \ \ E [u_i x_i] =0, \ \ u_i:= r_i+\epsilon_i.
$$
The least squares estimator of $\beta$ is
\begin{equation}\label{eq: original problem}
\widehat \beta := \arg \min_{b \in \Bbb{R}^k} \En \left[(y_i - p_i'b)^2\right]=\widehat{Q}^{-1}\En[p_i y_i]
\end{equation}
where $\widehat Q:=\En[p_i p_i']$.
The least squares estimator $\widehat{\beta}$ induces the estimator $\widehat g(x) := p(x)'\widehat \beta$ for the target function $g(x)$. Then it follows from (\ref{eq: approximation error}) that we can decompose the error in estimating the target function as
$$
\widehat g(x) - g(x) = p(x)'(\widehat \beta - \beta) - r(x),
$$
where the first term on the right-hand side is the estimation error and the second term is the approximation error.


We are also interested in various linear functionals $\theta$ of the conditional mean function. As discussed in the introduction, examples include
 the partial derivative function,  the average partial derivative function,  and the conditional average partial derivative. 
Importantly,  in each example above we could be interested in estimating $\theta=\theta(w)$ simultaneously for many values $w\in\I$. 
By the linearity of the series approximations, the above parameters can be seen as linear functions of the least squares coefficients $\beta$ up to an approximation error, that is
\begin{equation}\label{eq: cool}
\theta(w) = \ell_\theta(w)' \beta +  r_\theta(w),  \ \   w \in \I,
\end{equation}
where $\ell_\theta(w)' \beta$ is the series approximation, with
$\ell_\theta(w)$ denoting the $k$-vector of loadings on the coefficients,
and $r_\theta(w)$ is the remainder term, which corresponds to the
approximation error.  Indeed, the
decomposition (\ref{eq: cool}) arises from the application of different linear
operators $\mathcal{A}$ to the decomposition $g(\cdot) =
p(\cdot)'\beta + r(\cdot)$ and evaluating the resulting
functions at $w$:
\begin{equation}\label{eq: uncool}
\(\mathcal{A} g(\cdot)\)[w] = \(\mathcal{A} p(\cdot)\)[w]'\beta +  \(\mathcal{A} r(\cdot)\)[w].
\end{equation}

Examples of the operator $\mathcal{A}$ corresponding to the cases enumerated in the introduction are given by,
respectively,
\begin{itemize}
\item[1.]  a differential operator:  $(\mathcal{A}f) [x]= (\partial_j f) [x] $, so that
$$\ell_\theta(x) =\partial_j p(x), \ \ \  r_\theta(x) = \partial_j r(x);$$
\item[2.]  an integro-differential operator:  $\mathcal{A} f=\int \partial_j f(x) d\mu(x)$, so that
$$\ell_\theta  = \int \partial_j p(x) d \mu(x), \ \ \  r_\theta = \int \partial_j r(x) d \mu(x);  $$
\item[3.]  a partial integro-differential operator: $(\mathcal{A}f) [x_{2}]=\int \partial_j f(x) d\mu(x|x^s)$, so that
$$\ell_\theta(x^s) =\int \partial_j p(x)d\mu(x|x^s), \ \ \ r_\theta(x^s) =  \int \partial_j r(x) d\mu(x|x^s),$$
\end{itemize}
where $x^s$ is a subvector of $x$. For notational convenience, we use the formulation (\ref{eq: cool}) in the analysis, instead of the motivational formulation (\ref{eq: uncool}).

We shall provide the inference tools
that will be valid for inference on the series approximation
$$
\ell_\theta(w)' \beta, \ \ w \in \I.
$$
If the approximation error $r_\theta(w),  \ w  \in \I,$ is small enough as compared to the estimation error,
these tools will also be valid for inference on the  functional of interest
$$
\theta(w), \ \ w  \in \I.
$$
In this case,  the series approximation $\ell_\theta(w)$ is an important intermediary
target, whereas the functional $\theta(w)$ is the ultimate target. The
inference will be based on the plug-in estimator $\widehat
\theta(w) := \ell_\theta(w)' \widehat \beta$ of the the series
approximation $\ell_\theta(w)' \beta$ and hence of the final target
$\theta(w)$.

\section{Approximation Properties of Least Squares}\label{Sec:Approx}
Next we consider approximation properties of the least squares estimator. Not surprisingly, approximation properties must rely on the particular choice of approximating functions. At this point it is instructive to consider particular examples of relevant bases used in the literature. For each example, we state a bound on the following quantity:
$$
\xi_k:= \sup_{x \in \X} \| p(x)\|.
$$
This quantity will play a key role in our analysis.\footnote{{ Most results extend directly to the case that $\xi_k\geq  \max_{i\leq n} \| p(x_i)\|$ holds with probability $1-o(1)$. We refer to \cite{Hansen2014} for recent results that explicit allows for unbounded regressors which required extending the concentration inequalities for matrices.}}
 Excellent reviews of approximating properties of different series can also be found in \cite{huang1998} and \cite{Chen2006}, where additional references are provided.

\begin{example}[Polynomial series]\label{Ex:First} Let $\X=[0,1]$ and consider a polynomial series given by
$$
\widetilde p(x) = (1, x, x^2, ..., x^{k-1})'.
$$
In order to reduce collinearity problems, it is useful to orthonormalize the polynomial series with respect to the Lebesgue measure on $[0,1]$
to get the Legendre polynomial series
$$
p(x) = (1,\sqrt{3}x, \sqrt{5/4} (3 x^2 - 1),... )'.
$$
The Legendre polynomial series satisfies
$$
\xi_k \lesssim k;
$$
see, for example, \cite{Newey1997}. \qed
\end{example}

\begin{example}[Fourier series] Let $\X = [0,1]$ and consider a Fourier series given by
$$
p(x) = (1, \cos(2\pi j x), \sin(2\pi j x), j=1,2,..., (k-1)/2)',
$$
for $k$ odd. Fourier series is orthonormal with respect to the Lebesgue measure on $[0,1]$ and satisfies
$$
\xi_k \lesssim \sqrt{k},
$$
which follows trivially from the fact that every element of $p(x)$ is bounded in absolute value by one.\qed
\end{example}

\begin{example}[Spline series]  Let $\X=[0,1]$ and consider the linear regression spline series, or regression spline series of order 1, with a finite number of
equally spaced knots $l_1, \ldots, l_{k-2}$ in $\X$:
$$
\widetilde p(x) = (1, x,  (x-l_1)_+ ,\dots, (x-l_{k-2})_+)',
$$
or consider the cubic regression spline series, or regression spline series of order 3, with a finite number of equally spaced knots $l_1,\dots,l_{k-4}$:
$$
\widetilde p(x) = ( 1, x, x^2, x^3,  (x-l_1)^3_+,
 ..., (x-l_{k-4})^3_+)'.
$$
Similarly, one can define the regression spline series of any order $s_0$ (here $s_0$ is a nonnegative integer).
The function $x \mapsto \widetilde p(x)'b$ constructed using regression splines of order $s_0$ is
$s_0-1$ times continuously differentiable in $x$ for any $b$. Instead of regression splines,
it is often helpful to consider B-splines $p(x)=(p_1(x),\dots,p_k(x))'$, which are
linear transformations of the regression splines with lower multicollinearity; see \cite{DeBoor01} for the introduction to the theory of splines. B-splines are local in the sense that each B-spline $p_j(x)$ is supported on the interval $[l_{j(1)},l_{j(2)}]$ for some $j(1)$ and $j(2)$ satisfying $j(2)-j(1)\lesssim 1$ and there is at most $s_0+1$ non-zero B-splines on each interval $[l_{j-1},l_j]$. From this property of B-splines, it is easy to see that B-spline series satisfies
$$
\xi_k \lesssim \sqrt{k};
$$
see, for example, \cite{Newey1997}. \qed
\end{example}

\begin{example}[Cohen-Deubechies-Vial wavelet series]
Let $\X = [0,1]$ and consider Cohen-Deubechies-Vial (CDV) wavelet bases; see Section 4 in \cite{CDV93}, Chapter 7.5 in \cite{M09}, and Chapter 7 and Appendix B in \cite{Johnstone2011} for details on CDV wavelet bases.
CDV wavelet bases is a class of orthonormal with respect to the Lebesgue measure on $[0,1]$ bases.
Each such basis is built from a Daubechies scaling function $\phi$ (defined on $\Bbb{R}$) and the wavelet $\psi$ of order $s_0$ starting from a fixed resolution level $J_{0}$ such that $2^{J_{0}} \geq 2s_0$. The functions $\phi$ and $\psi$ are supported on $[0,2s_0-1]$ and $[-s_0+1,s_0]$, respectively. Translate $\phi$ so that it has the support $[-s_0+1,s_0]$.
 Let
\begin{equation*}
\phi_{l,m} (x) = 2^{l/2}\phi (2^{l}x - m), \ \psi_{l,m}(x) = 2^{l/2}\psi (2^{l} x - m), \ l,m \geq 0.
\end{equation*}
Then we can create the CDV wavelet basis from these functions as follows.
Take all the functions $\phi_{J_{0},m}, \psi_{l,m}$, $l\geq J_0$, that are supported in the interior of $[0,1]$ (these are functions $\phi_{J_{0},m}$ with $m=s_0-1,\dots,2^{J_{0}}-s_0$ and $\psi_{l,m}$ with $m=s_0-1,\dots,2^{l}-s_0, l \geq J_{0}$). Denote these functions $\widetilde{\phi}_{J_0,m}$, $\widetilde{\psi}_{l,m}$. To this set of functions, add suitable boundary corrected functions $\widetilde{\phi}_{J_0,0},\ldots,\widetilde{\phi}_{J_0,s_0-2}$, $\widetilde{\phi}_{J_0,2^{J_0}-s_0+1},\ldots,\widetilde{\phi}_{J_0,2^{J_0}-1}$,
$\widetilde{\psi}_{l,0},\ldots,\widetilde{\psi}_{l,s_0-2}$,
$\widetilde{\psi}_{l,2^{J_0}-s_0+1},\ldots,\widetilde{\psi}_{l,2^{J_0}-1}$,
$l\geq J_0$, so that
$\{ \widetilde{\phi}_{J_{0},m} \}_{0\leq m<2^{J_0}} \cup \{ \widetilde{\psi}_{l,m} \}_{0 \leq m <2^{l}, l \geq J_{0}}$ forms an orthonormal basis of $L^{2}[0,1]$. Suppose that $k = 2^{J}$ for some $J > J_{0}$. Then the CDV series takes the form:
\begin{equation*}
p(x) = (\widetilde{\phi}_{J_{0},0}(x),\ldots,\widetilde{\phi}_{J_{0},2^{J_{0}}-1}(x),\widetilde{\psi}_{J_{0},0}(x),\ldots,\widetilde{\psi}_{J-1,2^{J-1}-1}(x))'.
\end{equation*}
This series satisfies
\begin{equation*}
\xi_{k} \lesssim \sqrt{k}.
\end{equation*}
This bound can be derived by the same argument as that for B-splines \citep[see, for example,][Lemma 1 (i) for its proof]{Kato2013}. CDV wavelet bases is a flexible tool to approximate many different function classes. See, for example, \cite{Johnstone2011}, Appendix B.
\qed
\end{example}

\begin{example}[Local polynomial partition series]\label{eq: local polynomial}
Let $\mathcal{X}=[0,1]$ and define a local polynomial partition series as follows. Let $s_0$ be a nonnegative integer. Partition $\mathcal{X}$ as $0=l_0<l_1,\dots<l_{\widetilde{k}-1}<l_{\widetilde{k}}=1$ where $\widetilde{k}:=[k/(s_0+1)]+1$ where $[ a ]$ is the largest integer that is strictly smaller than $a$. For $j=1,\dots,\widetilde{k}$, define $\delta_j:[0,1]\to\{0,1\}$ by $\delta_j(x)=1$ if $x\in(l_{j-1},l_j]$ and $0$ otherwise. For $j=1,\dots,k$, define
$$
\widetilde{p}_j(x):=\delta_{[ j/(s_0+1)]+1}(x)x^{j-1-(s_0+1)[ j/(s_0+1)]}
$$
for all $x\in\mathcal{X}$. Finally, define the local polynomial partition series $p_1(\cdot),\dots,p_k(\cdot)$ of order $s_0$ as an orthonormalization of $\widetilde{p}_1(\cdot),\dots,\widetilde{p}_k(\cdot)$ with respect to the Lebesgue (or some other) measure on $\mathcal{X}$. The local polynomial partition series estimator was analyzed in detail in \cite{CF2013}. Its properties are somewhat similar to those of local polynomial estimator of \cite{Stone1982}. When the partition $l_0,\dots,l_{\widetilde{k}}$ satisfies $l_j-l_{j-1}\asymp 1/\widetilde{k}$, that is there exist constants $c,C>0$ independent of $n$ and such that $c/\widetilde{k}\leq l_j-l_{j-1}\leq C/\widetilde{k}$ for all $j=1,\dots,\widetilde{k}$, and the Lebesgue measure is used, the local polynomial partition series satisfies
$$
\xi_k\lesssim \sqrt{k}.
$$
This bound can be derived by the same argument as that for B-splines.\qed
\end{example}

\begin{example}[Tensor Products]\label{Ex:Last} Generalizations to multiple covariates are straightforward
using tensor products of unidimensional series. Suppose that the basic regressors
are
$$
x_i= (x_{1i}, ..., x_{di})'.
$$
Then we can create $d$
series for each basic regressor. Then we take all
interactions of functions from these $d$ series, called tensor
products, and collect them into a vector of regressors
$p_i$. If each series for a basic regressor has $J$
terms, then the final regressor has dimension $$k
= J^d,$$ which explodes exponentially in the
dimension $d$.  The bounds on $\xi_k$ in terms of $k$
remain the same as in one-dimensional case.\qed
\end{example}

Each basis described in Examples \ref{Ex:First}-\ref{Ex:Last} has different approximation properties which also depend on the particular class of functions $\mathcal{G}$. The following assumption captures the essence of this dependence into two quantities.

\textbf{Condition A.3} (Approximation) \textit{For each $n$ and $k$, there are finite constants  $c_k$ and $\ell_k$ such that for each $f \in \G$,}
$$
\|r_f\|_{F,2} := \sqrt{{ \begin{array}{l}\int_{x\in\mathcal{X}}\end{array} r_f^2(x) d F(x)}} \leq c_k
\ \ \
\mbox{{\em and}} \ \ \
\|r_f\|_{F,\infty} := \sup_{x \in \X}  |r_f(x)| \leq \ell_k c_k.
$$

Here $r_f$ is defined by (\ref{eq: beta g def}) and (\ref{eq: approximation error}) with $g$ replaced by $f$.
We call  $\ell_k$ the Lebesgue factor because of its relation to the Lebesgue constant defined in Section \ref{sub: lebesgue factor} below.
Together $c_k$ and $\ell_k$ characterize the approximation properties of the underlying class of functions under $L^2(\mathcal{X},F)$ and uniform distances. Note that constants $c_k=c_k(\mathcal{G})$ and $\ell_k=\ell_k(\mathcal{G})$ are allowed to depend $n$ but we omit indexing by $n$ for simplicity of notation. Next we discuss primitive bounds on $c_k$ and $\ell_k$.

\subsection{Bounds on $c_k$}
In what follows, we call the case where $c_k \to 0$ as $k \to \infty$ the correctly specified case. In particular, if the series are formed from bases that span  $\G$, then $c_k \to 0$ as $k \to  \infty$. However, if series are formed from bases that do not span  $\G$, then $c_k \not\to 0$ as $k \to  \infty$.  We call any case where $c_k \not \to 0$ the incorrectly specified (misspecified) case.

To give an example of the misspecified case, suppose that $d=2$, so that $x=(x_1,x_2)'$ and $g(x)=g(x_1,x_2)$. Further, suppose that the researcher mistakenly assumes that $g(x)$ is additively separable in $x_1$ and $x_2$: $g(x_1,x_2)=g_1(x_1)+g(x_2)$. Given this assumption, the researcher forms the vector of approximating functions $p(x_1,x_2)$ such that each component of this vector depends either on $x_1$ or $x_2$ but not on both; see \cite{Newey1997} and \cite{NPV99} for the description of nonparametric series estimators of separately additive models. Then note that if the true function $g(x_1,x_2)$ is not separately additive, linear combinations $p(x_1,x_2)'b$ will not be able to accurately approximate $g(x_1,x_2)$ for any $b$, so that $c_k$ does not converge to zero as $k\to\infty$. Since analysis of misspecified models plays an important role in econometrics, we include results both for correctly and incorrectly specified models.

To provide a bound on $c_k$, note that for any $f\in\mathcal{G}$,
$$
\inf_b \| f - p'b\|_{F,2}\leq \inf_{b} \| f - p'b\|_{F, \infty},
$$
so that it suffices to set $c_k$ such that $c_k\geq\sup_{f\in\mathcal{G}}\inf_{b} \| f - p'b\|_{F, \infty}$. Next, the bounds for $\inf_{b} \| f - p'b\|_{F, \infty}$ are readily available from the Approximation Theory; see \cite{DeVoreLorentz1993}. A typical example is based on the concept of $s$-smooth classes, namely H\"{o}lder classes of smoothness order $s$, $\Sigma_s(\mathcal{X})$. For $s\in(0,1]$, the H\"{o}lder class of smoothness order $s$, $\Sigma_s(\mathcal{X})$, is defined as the set of all functions $f:\mathcal{X}\to\mathbb{R}$ such that for $C>0$,
$$
|f(x)-f(\widetilde x)|\leq C\Big(\sum_{j=1}^d(x_j-\widetilde{x}_j)^2\Big)^{s/2}
$$
for all $x=(x_1,\dots,x_d)'$ and $\widetilde x=(\widetilde{x}_1,\dots,\widetilde{x}_d)'$ in $\mathcal{X}$. The smallest $C$ satisfying this inequality defines a norm of $f$ in $\Sigma_s(\mathcal{X})$, which we denote by $\|f\|_s$. For $s>1$, $\Sigma_s(\mathcal{X})$ can be defined as follows. For a $d$-tuple $\alpha=(\alpha_1,\dots,\alpha_d)$ of nonnegative integers, let
$$
D^\alpha=\partial_{x_1}^{\alpha_1}\dots\partial_{x_d}^{\alpha_d}.
$$
Let $[s]$ denote the largest integer strictly smaller than $s$. Then $\Sigma_s(\mathcal{X})$ is defined as the set of all functions $f:\mathcal{X}\to\mathbb{R}$ such that $f$ is $[s]$ times continuously differentiable and for some $C>0$,
$$
|D^{\alpha}f(x)-D^{\alpha}f(\widetilde{x})|\leq C\Big(\sum_{j=1}^d(x_j-\widetilde{x}_j)^2\Big)^{(s-[s])/2}\text{ and }|D^\beta f(x)|\leq C
$$
for all $x=(x_1,\dots,x_d)'$ and $\widetilde{x}=(\widetilde{x}_1,\dots,\widetilde{x}_d)'$ in $\mathcal{X}$ and for all $d$-tuples $\alpha=(\alpha_1,\dots,\alpha_d)$ and $\beta=(\beta_1,\dots,\beta_d)$ of nonnegative integers satisfying $\alpha_1+\dots+\alpha_d=[s]$ and $\beta_1+\dots+\beta_d\leq[s]$. Again, the smallest $C$ satisfying these inequalities defines a norm of $f$ in $\Sigma_s(\mathcal{X})$, which we denote $\|f\|_s$.

If $\G$ is a set of functions $f$ in $\Sigma_s(\mathcal{X})$ such that $\|f\|_s$ is bounded from above uniformly over all $f\in\G$ (that is, $\G$ is contained in a ball in $\Sigma_s(\mathcal{X})$ of finite radius), then we can take
\begin{equation}\label{eq: approximation properties standard}
c_k \lesssim  k^{-s/d}
\end{equation}
for the polynomial series and
$$
c_k \lesssim  k^{-(s\wedge s_0)/d}
$$
for spline, CDV wavelet, and local polynomial partition series of order $s_0$. If in addition we assume that each element of $\G$ can be extended to a periodic function, then (\ref{eq: approximation properties standard}) also holds for the Fourier series. See, for example, \cite{Newey1997} and \cite{Chen2006} for references.


\subsection{Bounds on $\ell_k$}\label{sub: lebesgue factor}
We say that a least squares approximation by a particular series for the function class
$\mathcal{G}$ is co-minimal if the Lebesgue factor $\ell_k$ is small in the sense
of being a slowly varying function in $k$. A simple bound on $\ell_k$, which is independent of $\G$, is established in the following proposition:

\begin{proposition}\label{lem: simple bound on l}
If $c_k$ is chosen so that $c_k\geq \sup_{f\in\G}\inf_b\|f-p'b\|_{F,\infty}$, then Condition A.3 holds with
$$
\ell_k\leq 1+\xi_k.
$$
\end{proposition}
The proof of this proposition is based on the ideas of \cite{Newey1997} and is provided in the Appendix.
The advantage of the bound established in this proposition is that it is universally applicable. It is, however, not sharp in many cases because $\xi_k$ satisfies
$$
\xi_k^2\geq E[\|p(x_i)\|^2]=E[p(x_i)'p(x_i)]=k
$$
so that $\xi_k  \gtrsim \sqrt{k}$ in all cases. Much sharper bounds follow from Approximation Theory for some important cases. To apply these bounds, define the Lebesgue constant:
$$
\widetilde{\ell}_k:=\sup \left (  \frac{\|p'\beta_f\|_{F,\infty}}{\|f\|_{F,\infty} }:  \|f\|_{F,\infty} \neq 0,  f \in \bar{\mathcal{G}} \right),
$$
where $\bar{\mathcal{G}} = \mathcal{G} + \{p'b: b \in \Bbb{R}^k\} = \{ f + p'b: f \in \mathcal{G}, b \in \Bbb{R}^k\}$.  The following proposition provides a bound on $\ell_k$ in terms of $\widetilde{\ell}_k$:
\begin{proposition}\label{lem: approximation theory bound}
If $c_k$ is chosen so that $c_k\geq\sup_{f\in\G}\inf_b\|f-p'b\|_{F,\infty}$, then Condition A.3 holds with
$$
\ell_k=1+\widetilde{\ell}_k.
$$
\end{proposition}
Note that in all examples above, we provided $c_k$ such that $c_k\geq\sup_{f\in\G}\inf_b\|f-p'b\|_{F,\infty}$, and so the results of Propositions \ref{lem: simple bound on l} and \ref{lem: approximation theory bound} apply in our examples.
We now provide bounds on $\widetilde{\ell}_k$.

\begin{example}[Fourier series, continued]\label{ex: fourier} For Fourier series on $\X=[0,1]$, $F = U(0,1)$, and $\G \subset C(\X)$
$$
\widetilde{\ell}_k \leq  C_0 \log k + C_1,
$$
where here and below $C_0$ and $C_1$ are some universal constants; see \cite{Zygmund2002}.\qed
\end{example}

\begin{example}[Spline series, continued] For continuous B-spline series on $\X=[0,1]$, $F=U(0,1)$, and $\G \subset C(\X)$
$$
\widetilde{\ell}_k \leq C_0,
$$
under approximately uniform placement of knots; see \cite{Huang2003b}. In fact, the result of Huang states that $\widetilde{\ell}_k \leq C$ whenever $F$ has the pdf on $[0,1]$ bounded from above by $\bar{a}$ and below from zero by $\underline{a}$ where $C$ is a constant that depends only on $\underline{a}$ and $\bar{a}$.\qed
\end{example}

\begin{example}[Wavelet series, continued] For continuous CDV wavelet series on $\X=[0,1]$, $F=U(0,1)$, and $\G \subset C(\X)$
$$
\widetilde{\ell}_k \leq C_0.
$$
The proof of this result was recently obtained by \cite{CC2013} who extended the argument of \cite{Huang2003b} for B-splines to cover wavelets. In fact, the result of Chen and Christensen also shows that $\widetilde{\ell}_k \leq C$ whenever $F$ has the pdf on $[0,1]$ bounded from above by $\bar{a}$ and below from zero by $\underline{a}$ where $C$ is a constant that depends only on $\underline{a}$ and $\bar{a}$.\qed
\end{example}

\begin{example}[Local polynomial partition series, continued]
For local polynomial partition series on $\mathcal{X}$, $F=U(0,1)$, and $\G\subset C(\X)$,
$$
\widetilde{\ell_k}\leq C_0.
$$
To prove this bound, note that first order conditions imply that for any $f\in\bar{\G}$,
$$
\beta_f=Q^{-1}E[p(x_1)f(x_1)]=E[p(x_1)f(x_1)].
$$
Hence, for any $x\in\X$,
$$
|p(x)'\beta_f|=|E[p(x)'p(x_1)f(x_1)]|\lesssim \|f\|_{F,\infty}
$$
where the last inequality follows by noting that the sum $p(x)'p(x_1)=\sum_{j=1}^kp_j(x)p_j(x_1)$ contains at most $s_0+1$ nonzero terms, all nonzero terms in the sum are bounded by $\xi_k^2\lesssim k$, and $p(x)'p(x_1)=0$ outside of a set with probability bounded from above by $1/k$ up to a constant. The bound follows.    Moreover, the bound $\widetilde{\ell}_k \leq C$  continues to hold whenever $F$ has the pdf on $[0,1]$ bounded from above by $\bar{a}$ and below from zero by $\underline{a}$ where $C$ is a constant that depends only on $\underline{a}$ and $\bar{a}$.
\qed
\end{example}

\begin{example}[Polynomial series, continued]  For Chebyshev polynomials with $\X=[0,1]$, $d F(x)/dx = 1/\sqrt{1-x^2}$,  and $\G \subset C(\X)$
$$
\widetilde{\ell}_k \leq C_0 \log k + C_1.
$$
This bound follows from a trigonometric representation of Chebyshev polynomials (see, for example, \cite{DeVoreLorentz1993}) and Example \ref{ex: fourier}.\qed
\end{example}

\begin{example}[Legendre Polynomials]  For Legendre polynomials that form an orthonormal basis on  $\X=[0,1]$ with respect to $F=(0,1)$,  and $\G = C(\X)$
$$
\widetilde{\ell}_k \geq C_0 \xi_k = C_1k,
$$
for some constants $C_0, C_1 >0$. See, for example, \cite{DeVoreLorentz1993}). This means that even though some series schemes generate well-behaved uniform approximations, others -- Legendre polynomials -- do not in general. However, the following example specifies ``tailored" function classes, for which Legendre and other series methods do automatically provide uniformly well-behaved approximations.  \qed
\end{example}


\begin{example}[Tailored Function Classes] For each type of series approximations, it is possible to specify function classes  for which the Lebesgue factors are constant or slowly varying with $k$.  Specifically, consider a collection
$$
\mathcal{G}_k = \left \{x \mapsto f(x) = p(x)'b + r(x):   \int r(x) p(x) d F(x) = 0,  \| r \|_{F,\infty} \leq \ell_k   \| r \|_{F,2}, \| r \|_{F,2} \leq c_k \right \},
$$
where $\ell_k \leq C$ or $\ell_k \leq C \log k$.  This example captures the idea, that for each type of series functions there are function classes that are well-approximated by this type. For example, Legendre polynomials may have poor Lebesgue factors in general, but there are well-defined function classes, where Legendre polynomials have well-behaved Lebesgue factors. This explains why polynomial approximations, for example, using Legendre polynomials, are frequently employed in empirical work.  We provide an empirically relevant example below, where polynomial approximation works just as well as a B-spline approximation.  In economic examples, both polynomial approximations and B-spline approximations are well-motivated if we consider them as more flexible forms of well-known, well-motivated functional forms in economics (for example, as more flexible versions of the linear-quadratic Mincer equations, or the more flexible versions of translog demand and production functions).  \qed 
\end{example}

The following example illustrate the performance of the series estimator using different bases for a real data set.

\begin{example}[Approximations of Conditional Expected Wage Function]
Here $g(x)$ is the mean of log wage ($y$)
conditional on education $$x \in \{8,9,10,11, 12, 13, 14, 16, 17,
18, 19, 20\}.$$ The function $g(x)$ is computed using population
data -- the 1990 Census data for the U.S. men of prime age; see \cite{ACF2006} for more details. So in this example, we know the true population function $g(x)$. We would like to know how well this
function is approximated when common approximation methods
are used to form the regressors. For simplicity we assume that $x_i$
is uniformly distributed (otherwise we can weigh by the frequency).
In population, least squares estimator solves the approximation problem: $\beta=\arg\min_b E[
\{g(x_i) - p_i'b\}^2] $ for $p_i =p(x_i)$, where we form $p(x)$ as (a)
linear spline (Figure 1, left) and (b) polynomial series (Figure 1, right), such that dimension of $p(x)$ is either $k=3$ or $k=8$. It is clear from these graphs that spline and polynomial series yield similar approximations.

\begin{figure}[vt]
\label{Fig:1}
\centering
\begin{tabular}{cc}
\includegraphics[scale=0.5]{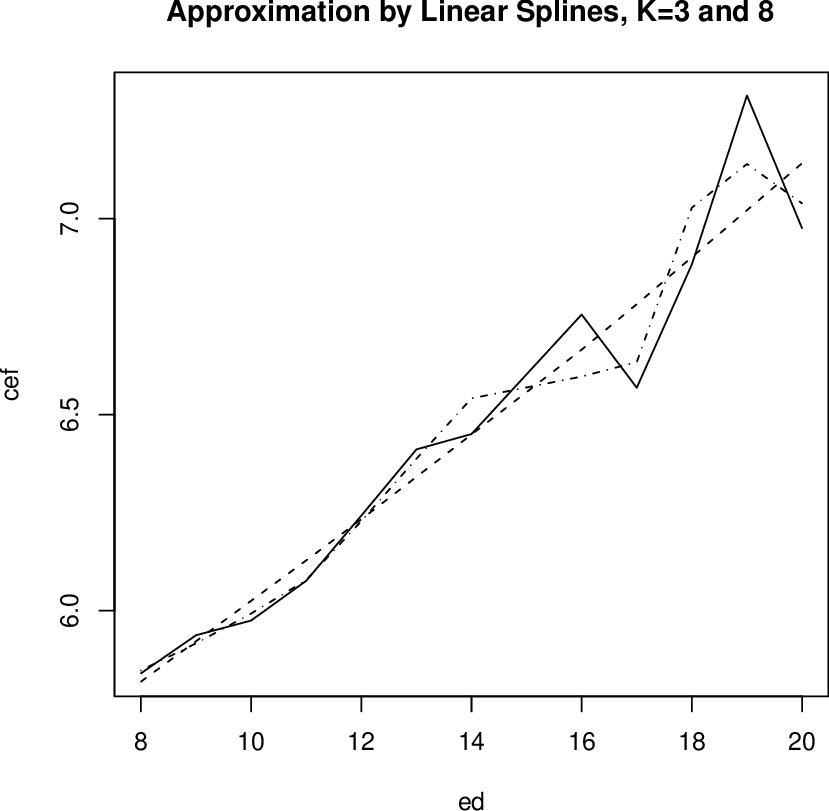}
&\includegraphics[scale=0.5]{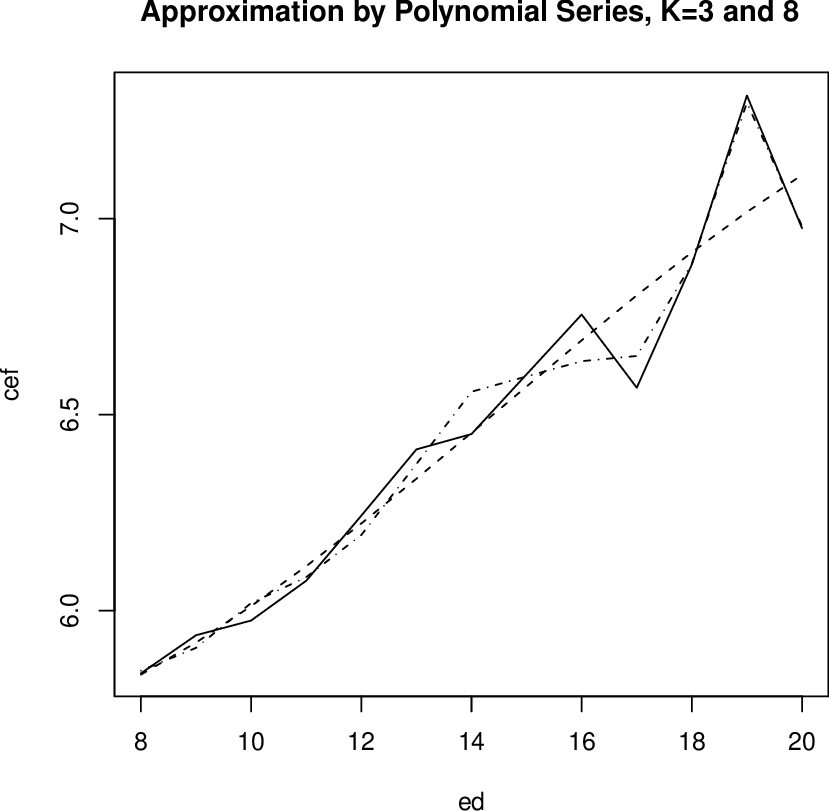}
\end{tabular}
\caption{\small{Conditional expectation function (cef) of $\log$ wage given education (ed) in the 1990 Census data for the U.S. men of prime age and its least squares approximation by spline (left panel) and polynomial series (right panel). Solid line - conditional expectation function; dashed line - approximation by $k=3$ series terms; dash-dot line - approximation by $k=8$ series terms}}
\end{figure}

In the table below, we also present $L^2$ and $L^\infty$ norms of approximating errors:

\begin{center}
\begin{tabular}{ c | c c c  c }
   &spline $k=3$ & spline $k=8$ & Poly $k=3$ & Poly $k=8$
  \\
  \hline
$L^2$ Error &  0.12 & 0.08 & 0.12 & 0.05 \\
$L^{\infty}$ Error &   0.29 & 0.17 & 0.30 & 0.12 \\
\hline
\end{tabular}
\end{center}
We see from the table that in this example, the Lebesgue factor, which is defined as the ratio of $L^\infty$ to $L^2$ errors, of the polynomial approximations is comparable to the Lebesgue factor of the spline approximations.
\qed
\end{example}

\section{Limit Theory}\label{Sec:LT}

\subsection{$L^{2}$ Limit Theory}\label{Sec:L2Rate}

After we have established the set-up, we proceed to derive our results. We start with a result on the $L^{2}$ rate of convergence. Recall that $\bar{\sigma}^2=\sup_{x\in\mathcal{X}}E[\epsilon_i^2|x_i=x]$. In the theorem below, we assume that $\overline{\sigma}^2\lesssim 1$. This is a mild regularity condition.

\begin{theorem}[$L^{2}$ rate of convergence]\label{theorem: L2 rate}  Assume that Conditions A.1-A.3 are satisfied. In addition, assume that $\xi_k^2\log k/n\to 0$ and $\overline{\sigma}^2\lesssim 1$. Then
under $c_k \to 0$,
\begin{equation}\label{eq: l2 rate 1}
\| \widehat g - g\|_{F,2} \lesssim_P  \sqrt{k/n} + c_k,
\end{equation}
and under $c_k \not\to 0$,
\begin{equation}\label{eq: l2 rate 2}
\|\widehat g - p'\beta\|_{F,2}\lesssim_P \sqrt{k/n}+(\ell_kc_k\sqrt{k/n})\wedge(\xi_kc_k/\sqrt{n}),
\end{equation}

\end{theorem}

\begin{remark}
(i) This is our first main result in this paper. The condition $\xi_k^2 \log k /n \to 0$, which we impose, weakens (hence generalizes) the conditions imposed in \cite{Newey1997} who required $k\xi_k^2/n \to 0$. For series satisfying $\xik \lesssim \sqrt{k}$, the condition $\xi_k^2 \log k /n \to 0$ amounts to
\begin{equation}\label{eq: superweak condition}
k \log k/n \to 0.
\end{equation}
This condition is the same as that imposed in \cite{Stone1994}, \cite{Huang2003}, and recently by \cite{CF2013} but the result (\ref{eq: l2 rate 1}) is obtained under the condition (\ref{eq: superweak condition}) in \cite{Stone1994} and \cite{Huang2003} only for spline series and in \cite{CF2013} only for local polynomial partition series. 
Therefore, our result improves on those in the literature by weakening the rate requirements on the growth of $k$ (with respect to $n$) and/or by allowing for a wider set of series functions.

(ii) Under the correct specification ($c_k\to 0$), the fastest $L^2$ rate of convergence is achieved by setting $k$ so that the approximation error and the sampling error are of the same order,
$$
\sqrt{k/n}  \asymp c_k.
$$
One consequence of this result is that for H\"{o}lder classes of smoothness order $s$, $\Sigma_s(\mathcal{X})$, with $c_k\lesssim k^{-s/d}$, we obtain the optimal $L^2$ rate of convergence by setting $k\asymp n^{d/(d+2s)}$, which is allowed under our conditions for all $s>0$ if $\xi_k\lesssim \sqrt{k}$ (Fourier, spline, wavelet, and local polynomial partition series). On the other hand, if $\xi_k$ is growing faster than $\sqrt{k}$, then it is not possible to achieve optimal $L^2$ rate of convergence for some $s>0$. For example, for polynomial series considered above, $\xi_k\lesssim k$, and so the condition $\xi_k^2\log k/n \to 0$ becomes $k^2\log k/n \to 0$. Hence, optimal $L^2$ rate of convergence is achieved by polynomial series only if $d/(d+2s)<1/2$ or, equivalently, $s>d/2$. Even though this condition is somewhat restrictive, it weakens the condition in \cite{Newey1997} who required $k^3/n\to 0$ for polynomial series, so that optimal $L^2$ rate in his analysis could be achieved only if $d/(d+2s)\leq 1/3$ or, equivalently, $s\geq d$. Therefore, our results allow to achieve optimal $L^2$ rate of convergence in a larger set of classes of functions for particular series.

(iii) The result (\ref{eq: l2 rate 2}) is concerned with the case when the model is misspecified ($c_k\not \to 0$). It shows that when $k/n\to 0$ and $(\ell_kc_k\sqrt{k/n})\wedge(\xi_kc_k/\sqrt{n})\to 0$, the estimator $\widehat{g}(\cdot)$ converges in $L^2$ to the surrogate function $p(\cdot)'\beta$ that provides the best linear approximation to the target function $g(\cdot)$. In this case, the estimator $\widehat{g}(\cdot)$ does not generally converge in $L^2$ to the target function $g(\cdot)$.
\qed
\end{remark}

\subsection{Pointwise Limit Theory}\label{Sec:Pointwise}

Next we focus on pointwise limit theory (some authors refer to pointwise limit theory as local asymptotics; see \cite{Huang2003b}). That is, we study asymptotic behavior of $\sqrt{n}\alpha'(\widehat{\beta}-\beta)$ and $\sqrt{n}(\widehat{g}(x)-g(x))$ for particular $\alpha\in S^{k-1}$ and $x\in\mathcal{X}$. Here $S^{k-1}$ denotes the space of vectors $\alpha$ in $\mathbb{R}^k$ with unit Euclidean norm: $\|\alpha\|=1$. Note that both $\alpha$ and $x$ implicitly depend on $n$. As we will show, pointwise results can be achieved under weak conditions similar to those we required in Theorem \ref{theorem: L2 rate}. The following lemma plays a key role in our asymptotic pointwise normality result.

\begin{lemma}[Pointwise Linearization]\label{lemma:linearization}  Assume that Conditions A.1-A.3 are satisfied. In addition, assume that $\xi_k^2\log k/n\to 0$ and $\overline{\sigma}^2\lesssim 1$. Then for any $\alpha \in S^{k-1}$,
 \begin{equation}\label{eq: lin1}
\sqrt{n} \alpha'( \widehat \beta - \beta) = \alpha' \Gn[ p_i (\epsilon_i + r_i)] + R_{1n}(\alpha),
 \end{equation}
where the term $R_{1n}(\alpha)$, summarizing the impact of unknown design, obeys
 \begin{equation}\label{eq: lin2}
R_{1n}(\alpha) \lesssim_P \sqrt{\frac{ \xik^2 \log k }{ n}}(1 + \sqrt{k} \ell_k c_k).
\end{equation}
Moreover,
 \begin{equation}\label{eq: lin3}
\sqrt{n} \alpha'( \widehat \beta - \beta) = \alpha' \Gn[ p_i \epsilon_i] + R_{1n}(\alpha) + R_{2n}(\alpha),
 \end{equation}
where the term $R_{2n}(\alpha)$, summarizing the impact of approximation error on the sampling error of the estimator, obeys
 \begin{equation}\label{eq: lin4}
R_{2n}(\alpha) \lesssim_P \ell_k c_k.
 \end{equation}
\end{lemma}

\begin{remark}
(i)  In summary, the only condition that generally matters for linearization (\ref{eq: lin1})-(\ref{eq: lin2}) is that $R_{1n}(\alpha) \to 0$, which holds if $\xi_k^2\log k/n\to 0$ and $k\xi_k^2\ell_k^2c_k^2\log k/n\to 0$. In particular, linearization (\ref{eq: lin1})-(\ref{eq: lin2}) allows for misspecification ($c_k\to 0$ is not required). 
In principle, linearization (\ref{eq: lin3})-(\ref{eq: lin4}) also allows for misspecification but the bounds are only useful if the model is correctly specified, so that $\ell_k c_k \to 0$. As in the theorem on $L^2$ rate of convergence, our main condition is that $\xi_k^2\log k/n \to 0$.

(ii) We conjecture that the bound on $R_{1n}(\alpha)$ can be improved for splines to
\begin{equation}\label{eq: lin5}
R_{1n}(\alpha) \lesssim_P \sqrt{\frac{ \xik^2 \log k }{ n}}(1 + \sqrt{\log k} \cdot \ell_k c_k).
\end{equation}
since it is attained by local polynomials and splines are also similarly localized.\qed
\end{remark}

With the help of Lemma \ref{lemma:linearization}, we derive our asymptotic pointwise normality result. We will use the following additional notation:
$$\widetilde \Omega:= Q^{-1} E [ (\epsilon_i + r_i)^2 p_i p_i'  ]Q^{-1}\text{ and }\Omega_0 := Q^{-1}E [\epsilon_i^2 p_i p_i'  ]Q^{-1}.
$$
In the theorem below, we will impose the condition that $\sup_{x \in \X} E\[\epsilon_i^2 1\{ |\epsilon_i|  > M \}|x_i =x\] \to 0$ as $M\to \infty$ uniformly over $n$. This is a mild uniform integrability condition. Specifically, it holds if for some $m>2$, $\sup_{x \in \mathcal{X}} E[|\epsilon_{i}|^{m} | x_{i} = x] \lesssim 1$. In addition, we will impose the condition that $1\lesssim \underline{\sigma}^2$. This condition is used to properly normalize the estimator.


\begin{theorem}[Pointwise Normality]\label{theorem: pointwise} Assume that Conditions A.1-A.3 are satisfied. In addition, assume that (i) $\sup_{x \in \X} E\[\epsilon_i^2 1\{ |\epsilon_i|  > M \}|x_i =x\] \to 0$ as $M\to \infty$ uniformly over $n$, (ii) $1\lesssim\underline{\sigma}^2$, and (iii) $(\xi_k^2\log k/n)^{1/2}(1+k^{1/2}\ell_kc_k)\to 0$. Then for any $\alpha \in S^{k-1}$,
\begin{equation}\label{eq: pointwise normality 1}
\sqrt{n} \frac{\alpha'(\widehat \beta - \beta)}{\|\alpha' \Omega^{1/2}\|} =_d  N(0,1) + o_P(1),
\end{equation}
where we set $\Omega =\widetilde \Omega$
but if $R_{2n}(\alpha) \to_P 0$, then we can set
$\Omega = \Omega_0$.
Moreover, for any $x \in \X$ and $s(x):= \Omega^{1/2}p(x)$,
\begin{equation}\label{eq: pointwise normality 2}
\sqrt{n} \frac{p(x)'(\widehat \beta - \beta)}{\|s(x)\|} =_d N(0,1) + o_P(1),
\end{equation}
and if  the approximation error  is negligible relative to the estimation error, namely $\sqrt{n} r(x)=o(\|s(x)\|)$, then
\begin{equation}\label{eq: pointwise normality 3}
\sqrt{n} \frac{\widehat g(x) - g(x)}{\|s(x)\|} =_d N(0,1) + o_P(1).
\end{equation}
\end{theorem}
\begin{remark}
(i) This is our second main result in this paper. The result delivers pointwise convergence in distribution for any sequences $\alpha=\alpha_n$ and $x=x_n$ with $\alpha\in S^{k-1}$ and $x\in\mathcal{X}$. In fact, the proof of the theorem implies that the convergence is uniform over all sequences. Note that the normalization
factor $\|s(x)\|$ is the pointwise standard error, and it is of a typical order
$ \|s(x)\| \propto \sqrt{k} $ at most points.
In this case the condition for negligibility of approximation error $\sqrt{n} r(x)/\|s(x)\|\to 0$, which can be understood as an undersmoothing condition, can be replaced by $$\sqrt{n/k}\cdot  \ell_k c_k\to 0.$$
When $\ell_k c_k\lesssim k^{-s/d}\log k$, which is often the case if $\G$ is contained in a ball in $\Sigma_s(\mathcal{X})$ of finite radius (see our examples in the previous section), this condition substantially weakens an assumption in \cite{Newey1997} who required $\sqrt{n}k^{-s/d} \to 0$ in a similar set-up. 

(ii) When applied to splines, our result is somewhat less sharp than that of \cite{Huang2003b}. Specifically, Huang required that $\xi_k^2\log k/n \to 0$ and $(n/k)^{1/2}\cdot\ell_k c_k\to 0$ whereas we require $(k\xi_k^2\log k/n)^{1/2}\ell_k c_k\to 0$ in addition to Huang's conditions (see condition (iii) of the theorem). The difference can likely be explained by the fact that we use linearization bound (\ref{eq: lin2}) whereas for splines it is likely that (\ref{eq: lin5}) holds as well.

(iii) More generally, our asymptotic pointwise normality result, as well as other related results in this paper, applies to any problem
where the estimator of  $g(x) = p(x)'\beta + r(x)$ takes the form $p(x)'\widehat \beta$, where $\widehat \beta$ admits linearization
of the form (\ref{eq: lin1})-(\ref{eq: lin4}).  \qed
\end{remark}

\subsection{Uniform Limit Theory}\label{Sec:Uniform}
Finally, we turn to a uniform limit theory. Not surprising, stronger conditions are required for our results to hold when compared to the pointwise case. Let $m>2$. We will need the following assumption on the tails of the regression errors.

\textbf{Condition A.4} (Disturbances) \textit{Regression errors satisfy $\sup_{x \in \mathcal{X}} E[|\epsilon_{i}|^{m} | x_{i} = x] \lesssim 1$.}

It will be convenient to denote $\alpha(x):=p(x)/\|p(x)\|$ in this subsection. Moreover, denote
$$
\xi_k^L:=\sup_{x,x'\in\mathcal{X}: \, x\neq x'}\frac{\|\alpha(x)-\alpha(x')\|}{\|x - x' \|}
$$
We will also need the following assumption on the basis functions to hold with the same $m>2$ as that in Condition A.4.

\textbf{Condition A.5} (Basis) \textit{Basis functions are such that (i) $\xi_{k}^{2m/(m-2)} \log k/n \lesssim 1$, (ii) $\log \xi^L_k\lesssim \log k$, and (iii) $\log \xi_k\lesssim \log k$.}

The following lemma provides uniform linearization of the series estimator and plays a key role in our derivation of the uniform rate of convergence.
\begin{lemma}[Uniform Linearization]\label{lemma: uniform linearization}  Assume that Conditions A.1-A.5 are satisfied. Then
\begin{equation}\label{eq: lin1U}
\sqrt{n} \alpha(x)'( \widehat \beta - \beta) = \alpha(x)' \mathbb{G}_n[ p_i (\epsilon_i + r_i)] + R_{1n}(\alpha(x)),
\end{equation}
where  $R_{1n}(\alpha(x))$, summarizing the impact of unknown design, obeys
\begin{equation}\label{eq: lin2U}
R_{1n}(\alpha(x)) \lesssim_P \sqrt{\frac{ \xi_k^2 \log k }{n}} (n^{1/m} \sqrt{\log k}  + \sqrt{k} \cdot \ell_kc_{k})=:\bar{R}_{1n}
\end{equation}
uniformly over $x \in \mathcal{X}$.
Moreover,
\begin{equation}\label{eq: lin3U}
\sqrt{n} \alpha(x)'( \widehat \beta - \beta) = \alpha(x)' \mathbb{G}_n[ p_i \epsilon_i ] + R_{1n}(\alpha(x)) +R_{2n}(\alpha(x)),
\end{equation}
where  $R_{2n}(\alpha(x))$, summarizing the impact of approximation error on the sampling error of the estimator, obeys
\begin{equation}\label{eq: lin4U}
R_{2n}(\alpha(x)) \lesssim_P \sqrt{\log{k}} \cdot \ell_kc_{k}=:\bar{R}_{2n}
\end{equation}
uniformly over $x\in\mathcal{X}$.
\end{lemma}
\begin{remark}
As in the case of pointwise linearization, our results on uniform linearization (\ref{eq: lin1U})-(\ref{eq: lin2U}) allow for misspecification ($c_k\to 0$ is not required). In principle, linearization (\ref{eq: lin3U})-(\ref{eq: lin4U}) also allows for misspecification but the bounds are most useful if the model is correctly specified so that $(\log k)^{1/2}\ell_k c_k \to 0$. We are not aware of any similar uniform linearization result in the literature. We believe that this result is useful in a variety of problems. Below we use this result to derive good uniform rate of convergence of the series estimator. Another application of this result would be in testing shape restrictions in the nonparametric model. \qed
\end{remark}
The following theorem provides uniform rate of convergence of the series estimator.

%

\begin{theorem}[Uniform Rate of Convergence]\label{theorem: uniform rate}
Assume that Conditions A.1-A.5 are satisfied. Then
\begin{equation}\label{eq: 4.13}
\sup_{x \in \mathcal{X}} |\alpha(x)' \mathbb{G}_n[ p_i \epsilon_i]| \lesssim_P \sqrt{\log k}.
\end{equation}
Moreover, for $\bar{R}_{1n}$ and $\bar{R}_{2n}$ given above we have
\begin{equation}\label{eq: uniform rate 2}
\sup_{x \in \mathcal{X}} |p(x)'( \widehat \beta - \beta)| \lesssim_P \frac{\xi_k}{\sqrt{n}}( \sqrt{\log k}  + \bar{R}_{1n} + \bar{R}_{2n})
\end{equation}
and
\begin{equation}\label{eq: uniform rate 3}
\sup_{x \in \mathcal{X}} |\widehat g(x) - g(x)| \lesssim_P \frac{\xi_k}{\sqrt{n}}( \sqrt{\log k}  + \bar{R}_{1n} + \bar{R}_{2n}) + \ell_kc_{k}.
\end{equation}
\end{theorem}
\begin{remark}This is our third main result in this paper. Assume that $\G$ is a ball in $\Sigma_s(\mathcal{X})$ of finite radius, $\ell_k c_k\lesssim k^{-s/d}$, $\xi_k\lesssim \sqrt{k}$, and $\bar{R}_{1n}+\bar{R}_{2n}\lesssim (\log k)^{1/2}$. Then the bound in (\ref{eq: uniform rate 3}) becomes
$$
\sup_{x \in \mathcal{X}} |\widehat g(x) - g(x)| \lesssim_P \sqrt{\frac{k\log k}{n}}+k^{-s/d}.
$$
Therefore, setting $k\asymp(\log n/n)^{-d/(2s+d)}$, we obtain
$$
\sup_{x \in \mathcal{X}} |\widehat g(x) - g(x)|  \lesssim_P \left(\frac{\log n}{n}\right)^{s/(2s+d)},
$$
which is the optimal uniform rate of convergence in the function class $\Sigma_s(\mathcal{X})$; see \cite{Stone1982}. To the best of our knowledge, our paper is the first to show that the series estimator attains the optimal uniform rate of convergence under these rather general conditions; see the next comment. We also note here that it has been known for a long time that a local polynomial (kernel) estimator achieves the same optimal uniform rate of convergence; see, for example, \cite{Tsybakov2003}, and it was also shown recently by \cite{CF2013} that local polynomial partition series estimator also achieves the same rate. { Recently, in an effort to relax the independence assumption, the working paper \cite{CC2013}, which appeared in ArXiv in 2013, approximately 1 year after our paper was posted to ArXiv and submitted for publication, \footnote{Our paper was submitted for publication and to ArXiv on December 3, 2012.  Our result as stated here did not change since the original submission.} derived similar uniform rate of convergence result allowing for $\beta$-mixing conditions, see their Theorem 4.1 for specific conditions.}\qed
 \end{remark}

\begin{remark} Primitive conditions leading to inequalities $\ell_k c_k\lesssim k^{-s/d}$ and $\xi_k\lesssim \sqrt{k}$ are discussed in the previous section. Also, under the assumption that $\ell_k c_k\lesssim k^{-s/d}$, inequality $\bar{R}_{2n}\lesssim (\log k)^{1/2}$ follows automatically from the definition of $\bar{R}_{2n}$.
Thus, one of the critical conditions to attain the optimal uniform rate of convergence is that we require $\bar{R}_{1n}\lesssim (\log k)^{1/2}$. Under our other assumptions, this condition holds if $k\log k/n^{1-2/m}\lesssim 1$ and $k^{2-2s/d}/n\lesssim 1$, and so we can set $k\asymp(\log n/n)^{-d/(2s+d)}$ if $d/(2s+d)<1-2/m$ and $(2d-2s)/(2s+d)<1$ or, equivalently, $m>2+d/s$ and $s/d>1/4$.
\qed
\end{remark}

After establishing the auxiliary results on the uniform rate of convergence, we present two results on inference based on the series estimator.
The first result on inference is concerned with the strong approximation of a series process by a Gaussian process and is a (relatively) minor extension of the result obtained by  \cite{CLR2008}. The extension is undertaken to allow for a non-vanishing specification error to cover misspecified models. In particular, we make a distinction between $\widetilde \Omega= Q^{-1} E [ (\epsilon_i + r_i)^2 p_i p_i'  ]Q^{-1},$
and $\Omega_0 = Q^{-1}E [\epsilon_i^2 p_i p_i'  ]Q^{-1}$ which are potentially asymptotically different if $\bar{R}_{2n} \not\to_P 0$.
To state the result, let $a_n$ be some sequence of positive numbers satisfying $a_n \to \infty$.

\begin{theorem}[Strong Approximation by a Gaussian Process]\label{theorem: uniform normality}
Assume that Conditions A.1-A.5 are satisfied with $m\geq 3$. In addition, assume that (i) $\bar{R}_{1n} = o_P(a_n^{-1})$, (ii) $1\lesssim \underline{\sigma}^2$,
and (iii)
$a_n^6  k^4 \xi^2_k  (1 + \ell_k^3 c_k^3)^2 \log^2 n/n  \to 0.$
Then for some $\mathcal{N}_k \sim N(0, I_k)$,
\begin{equation}\label{eq: strong approximation process 1}
\sqrt{n} \frac{\alpha(x)'(\widehat \beta - \beta)}{\|\alpha(x)' \Omega^{1/2}\|} =_d  \frac{\alpha(x)' \Omega^{1/2}}{\|\alpha(x)' \Omega^{1/2}\|}  \mathcal{N}_k + o_P(a_n^{-1}) \text{ in } \ell^{\infty}(\mathcal{X}),
\end{equation}
so that for $s(x)= \Omega^{1/2}p(x)$,
\begin{equation}\label{eq: strong approximation process 2}
 \sqrt{n} \frac{p(x)'(\widehat \beta - \beta)}{\|s(x)\|} =_d \frac{s(x)'}{\|s(x)\|} \mathcal{N}_k + o_P(a_n^{-1})  \text{ in } \ell^{\infty}(\X),
\end{equation}
and if  $\sup_{x \in \X}\sqrt{n} |r(x)|/\|s(x)\| =o(a_n^{-1}) $, then
\begin{equation}\label{eq: strong approximation process 3}
\sqrt{n} \frac{\widehat g(x) - g(x)}{\|s(x)\|} =_d \frac{s(x)'}{\|s(x)\|} \mathcal{N}_k + o_P(a_n^{-1})  \text{ in } \ell^{\infty}(\X),
\end{equation}
where we set $\Omega =\widetilde \Omega$ but if $\bar{R}_{2n} = o_P(a_n^{-1})$, then we can set
$\Omega = \Omega_0$.
\end{theorem}
\begin{remark}\label{com: on strong approximations}One might hope to have a result of the form
\begin{equation}\label{eq: ideal result 2}
\sqrt{n} \frac{\widehat g(x) - g(x)}{\|s(x)\|}\to_d G(x)  \text{ in } \ell^{\infty}(\X),
\end{equation}
where $\{G(x):x\in\mathcal{X}\}$ is some fixed zero-mean Gaussian process.    However, one can show that the process on the left-hand side of (\ref{eq: ideal result 2}) is not asymptotically equicontinuous, and so it does not have a limit distribution.
 Instead, Theorem \ref{theorem: uniform normality} provides an approximation of the series process by \textit{a sequence} of zero-mean Gaussian processes $\{G_k(x):x\in\mathcal{X}\}$
  $$
G_k(x):=   \frac{\alpha(x)' \Omega^{1/2}}{\|\alpha(x)' \Omega^{1/2}\|}  \mathcal{N}_k ,
  $$
with the stochastic error of size $o_P(a^{-1}_n)$.  Since $a_n\to\infty$, under our conditions the theorem implies that the series process is well approximated by a Gaussian process, and so the theorem can be interpreted as saying that in large samples, the distribution of the series process depends on the distribution of the data only via covariance matrix $\Omega$; hence, it allows us to perform inference based on the whole series process. Note that the conditions of the theorem are quite strong in terms of growth requirements on $k$, but the result of the theorem is also much stronger than the pointwise normality result: it asserts that the entire series process is uniformly close to a Gaussian process of the stated form.
\qed
\end{remark}


Our result on the strong approximation by a Gaussian process plays an important role in our second result on inference that is concerned with the weighted bootstrap. Consider a set of weights
$h_1,\ldots,h_n$ that are i.i.d. draws from the standard exponential
distribution and are independent of the data. For each draw of such weights, define the weighted
bootstrap draw of the least squares estimator as a solution to the least squares problem
weighted by $h_1,\ldots,h_n$, namely
\begin{equation}\label{eq: weighted bootstrap problem}
\widehat \beta^b \in \arg \min_{b \in \RR^k} \En[ h_i (y_i-p_i'b)^2].
\end{equation}
For all $x\in\mathcal{X}$, denote $\widehat{g}^b(x)=p(x)'\widehat{\beta}^b$. The following theorem establishes a new result that states that the weighted bootstrap
distribution is valid for approximating the distribution of the series process.

\begin{theorem}[Weighted Bootstrap Method]\label{Thm:MainBootstrap}
(1) Assume that Conditions A.1-A.5 are satisfied. In addition, assume that $(\xi_k(\log n)^{1/2})^{2m/(m-2)}\lesssim 1$. Then the weighted bootstrap process satisfies
$$
\sqrt{n} \alpha(x)'( \widehat \beta^b - \widehat\beta) = \alpha(x)' \Gn[ (h_i-1)p_i (\epsilon_i + r_i)] + R_{1n}^b(\alpha(x)),
$$
where $R_{1n}^b(\alpha(x))$ obeys
\begin{equation}\label{eq: weighting bootstrap approximation 1}
R_{1n}^b(\alpha(x))  \lesssim_P \sqrt{\frac{\xi_k^2\log^3 n}{n}}(n^{1/m}\sqrt{\log n} + \sqrt{k}\cdot \ell_kc_k) =:\bar{R}_{1n}^b
\end{equation}
uniformly over $x\in\mathcal{X}$.

(2) If, in addition, Conditions A.4 and A.5 are satisfied with $m\geq 3$ and (i) $\bar{R}_{1n}^b=o_P(a_n^{-1})$, (ii) $1\lesssim \underline{\sigma}^2$, and (iii) $a_n^6  k^4 \xi^2_k  (1 + \ell_k^3 c_k^3)^2 \log^2 n/n  \to 0$ hold, then for $s(x)= \Omega^{1/2}p(x)$ and some $\mathcal{N}_k\sim N(0,I_k)$,
\begin{equation}\label{eq: weighted bootstrap approximation 2}
 \sqrt{n} \frac{p(x)'(\widehat \beta^b - \widehat\beta)}{\|s(x)\|} =_d \frac{s(x)'}{\|s(x)\|} \mathcal{N}_k + o_P(a_n^{-1})  \text{ in } \ell^{\infty}(\X),
\end{equation}
and so
\begin{equation}\label{eq: weighted bootstrap approximation 3}
\sqrt{n} \frac{\widehat g^b(x) - \widehat{g}(x)}{\|s(x)\|} =_d \frac{s(x)'}{\|s(x)\|} \mathcal{N}_k + o_P(a_n^{-1})  \text{ in } \ell^{\infty}(\X).
\end{equation}
where we set $\Omega =\widetilde \Omega$, but if $\bar{R}_{2n} = o_P(a_n^{-1})$, then we can set
$\Omega = \Omega_0.$

(3) Moreover, the bounds (\ref{eq: weighting bootstrap approximation 1}), (\ref{eq: weighted bootstrap approximation 2}), and (\ref{eq: weighted bootstrap approximation 3}) continue to hold in $P$-probability if we replace the
unconditional probability $P$ by the conditional probability computed given the data, namely if we replace $P$ by $P^*(\cdot \mid D)$ where $D=\{(x_i,y_i):i=1,\dots,n\}$.

\end{theorem}
\begin{remark}
(i) This is our fourth main and new result in this paper. The theorem implies that the weighted bootstrap process can be approximated by a copy of the same Gaussian process as that used to approximate original series process.

(ii) We emphasize that the theorem does not require the correct specification, that is the case $c_k\not\to 0$ is allowed. 
Also, in this theorem, symbol $P$ refers to a joint probability measure with respect to the data $D=\{(x_i,y_i):i=1,\dots,n\}$ and the set of bootstrap weights $\{h_i:i=1,\dots,n\}$.\qed
\end{remark}

We close this section by establishing sufficient conditions for consistent estimation of $\Omega$. Recall that $Q=E[p_ip_i']=I$. In addition, denote $\Sigma = E[(\epsilon_i+r_i)^2p_ip_i']$, $\widehat Q = \En[p_ip_i']$, and $\widehat \Sigma = \En[\widehat \epsilon_i^2 p_ip_i']$ where $\widehat \epsilon_i = y_i-p_i'\widehat \beta$, and let $v_n =(E[\max_{1\leq i \leq n} | \epsilon_i|^2])^{1/2}$.

\begin{theorem}[Matrices Estimation]\label{thm:Matrices}
Assume that Conditions A.1-A.5 are satisfied. In addition, assume that $\bar{R}_{1n}+\bar{R}_{2n}\lesssim(\log k)^{1/2}$. Then
$$
\|\widehat Q - Q\| \lesssim_P \sqrt{\frac{\xi_k^2\log k}{n}}=o(1) \ \ \mbox{and} \ \  \|\widehat \Sigma - \Sigma\|\lesssim_P (v_n \vee 1+\ell_kc_k) \sqrt{\frac{\xi_k^2\log k}{n}}=o(1).
$$
Moreover, for $\widehat \Omega = \widehat Q^{-1}\widehat \Sigma\widehat Q^{-1}$ and $\Omega=Q^{-1}\Sigma Q^{-1}$,
$$
\|\widehat{\Omega} - \Omega\|\lesssim_P (v_n \vee 1+\ell_kc_k) \sqrt{\frac{\xi_k^2\log k}{n}}=o(1).
$$
\end{theorem}
\begin{remark}
Theorem \ref{thm:Matrices} allows for consistent estimation of the matrix $Q$ under the mild condition $\xi_k^2\log k/n \to 0$ and for consistent estimation of the matrices $\Sigma$ and $\Omega$ under somewhat more restricted conditions. Not surprisingly, the estimation of $\Sigma$ and $\Omega$ depends on the tail behavior of the error term via the value of $v_n$.
Note that under Condition A.4, we have that $v_n \lesssim n^{1/m}$. \qed
\end{remark}

\section{Rates and Inference on Linear Functionals}\label{Sec:LinearFunctionals}

In this section, we derive rates and inference results for linear functionals $\theta(w), w\in \I$ of the conditional expectation function such as its derivative, average derivative, or conditional average derivative. To a large extent, with the exception of Theorem \ref{thm: confidence bands validity}, the results presented in this section can be considered as an extension of results presented in Section \ref{Sec:LT}, and so similar comments can be applied as those given in Section \ref{Sec:LT}. Theorem \ref{thm: confidence bands validity} deals with construction of uniform confidence bands for linear functionals under weak conditions and is a new result.

By the linearity of the series approximations, the linear functionals can be seen as linear functions of the least squares coefficients $\beta$ up to an approximation error, that is
$$
\theta(w) = \ell_\theta(w)' \beta +  r_\theta(w),  \ \   w \in \I,
$$
where $\ell_\theta(w)' \beta$ is the series approximation, with
$\ell_\theta(w)$ denoting the $k$-vector of loadings on the coefficients,
and $r_\theta(w)$ is the remainder term, which corresponds to the
approximation error.
Throughout this section, we assume that $\mathcal{I}$ is a subset of some Euclidean space $\mathbb{R}^l$ equipped with its usual norm $\|\cdot\|$. We allow $\I=\mathcal{I}_n$ to depend on $n$ but for simplicity, we assume that the diameter of $\mathcal{I}$ is bounded from above uniformly over $n$. Results allowing for the case where $\mathcal{I}$ is expanding as $n$ grows can be covered as well with slightly more technicalities.

In order to perform inference, we construct estimators of
$ \sigma_\theta^2(w) = \ell_\theta(w)'\Omega\ell_\theta(w)/n$, the
variance of the associated linear functionals, as
\begin{equation}\label{Def:hatsigma2n}
\widehat \sigma_\theta^2(w) =  \ell_\theta(w)'  \widehat \Omega \ell_\theta(w)/n.
\end{equation}

In what follows, it will be convenient to have the following result on consistency of $\widehat{\sigma}_\theta(w)$:
\begin{lemma}[Variance Estimation for Linear Functionals]\label{lem: variance consistency}
Assume that Conditions A.1-A.5 are satisfied. In addition, assume that (i) $\bar{R}_{1n}+\bar{R}_{2n}\lesssim(\log k)^{1/2}$ and (ii) $1\lesssim \underline{\sigma}^2$. Then
$$
\left|\frac{\widehat{\sigma}_\theta(w)}{\sigma_\theta(w)}-1\right|\lesssim_P \|\widehat{\Omega}-\Omega\|\lesssim_P (v_n \vee 1+\ell_kc_k) \sqrt{\frac{\xi_k^2\log k}{n}}=o(1)
$$
uniformly over $w\in\mathcal{I}$.
\end{lemma}
By Lemma \ref{lem: variance consistency}, under our
conditions, (\ref{Def:hatsigma2n}) is uniformly consistent for
$\sigma_\theta^2(w)$ in the sense that $\widehat
\sigma_\theta^2(w)/ \sigma_\theta^2(w) = 1 + o_P(1)$ uniformly over $w\in \I$.

\subsection{Pointwise Limit Theory for Linear Functionals}\label{Sec:LinearPointwise}



We now present a result on pointwise rate of convergence for linear functionals. The rate we derive is $\|\ell_\theta(w)\|/\sqrt{n}$. Some examples with explicit bounds on $\|\ell_\theta(w)\|$ are given below.
\begin{theorem}[Pointwise Rate of Convergence for Linear Functionals]
\label{theorem: pointwise rate}  Assume that Conditions A.1-A.3 are satisfied. In addition, assume that (i) $\sqrt{n}|r_\theta(w)|/\|\ell_\theta(w)\| \to 0$,  (ii) $\bar{\sigma}^2\lesssim 1$, (iii) $(\xi_k^2\log k/n)^{1/2}(1+k^{1/2}\ell_kc_k) \to 0$, and (iv) $\ell_kc_k\to 0$. Then
$$
| \widehat \theta(w) -  \theta(w)| \lesssim_P  \frac{\|\ell_{\theta}(w)\|}{\sqrt{n}}.
$$
\end{theorem}
\begin{remark}\label{Remark:LinearFunctional}
(i) This theorem shows in particular that $\widehat{\theta}(w)$ is $\sqrt{n}$-consistent whenever $\|\ell_{\theta}(w)\|\lesssim 1$. A simple example of this case is $\theta=\theta(w)=E[g(x_1)]$. In this example, $\ell=\ell(w)=E[p(x_1)]$, and so $\|\ell\|=\|E[p(x_1)]\|\lesssim 1$ where the last inequality follows from the argument used in the proof of Proposition \ref{lem: simple bound on l}. Another simple example is $\theta=\theta(w)=E[p(x_1)g(x_1)]=\beta_1$. In this example, $\ell=\ell(w)$ is a $k$-vector whose first component is 1 and all other components are 0, and so $\|\ell\|\lesssim 1$. This example trivially implies $\sqrt{n}$-consistency of the series estimator of the linear part of the partially linear model. Yet another example, which is discussed in \cite{Newey1997}, is the average partial derivative.

(ii) Condition $\sqrt{n}|r_{\theta}(w)|/\|\ell_\theta(w)\|\to 0$ imposed in this theorem can be understood as undersmoothing condition. Unfortunately, to the best of our knowledge, there is no theoretically justified practical procedure in the literature that would lead to a desired level of undersmoothing. Some ad hoc suggestions include using cross validation or ``plug-in'' method to determine the number of series terms that would minimize the asymptotic integrated mean-square error of the series estimator \citep[see][]{Hardle1990} and then blow up the estimated number of series terms by some number that grows to infinity as the sample size increases.\qed
\end{remark}

To perform pointwise inference, we consider the t-statistic:
$$
t(w) = \frac{  \widehat \theta(w) - \theta(w) }{ \widehat \sigma_\theta(w)}.
$$
We can carry out standard inference based on this statistic because of the following theorem.

\begin{theorem}[Pointwise Inference for Linear Functionals]\label{Thm:DistributionsInferentialpointwise}
Assume that the conditions of Theorem \ref{theorem: pointwise}
and Lemma \ref{lem: variance consistency} are satisfied. In addition, assume that $\sqrt{n}|r_\theta(w)|/\|\ell_{\theta}(w)\|\to 0$. Then
$$
t(w) \to_d N(0,1).
$$
\end{theorem}
The same comments apply here as those given in Section \ref{Sec:Pointwise} for pointwise results on estimating the function $g$ itself.

\subsection{Uniform Limit Theory for Linear Functionals}\label{Sec:LinearUniform}
In obtaining uniform rates of convergence and inference results for linear functionals, we will denote
$$
\xi_{k,\theta}:=\sup_{w\in\I}\|\ell_{\theta}(w)\|\,\,\text{ and }\,\,\xi_{k,\theta}^L:=\sup_{w,w'\in\mathcal{I}: \, w\neq w'}\frac{\|\ell_{\theta}(w)-\ell_{\theta}(w')\|}{\|w - w' \|}.
$$

The value  of $\xi_{k,\theta}$ depends on the choice of the basis for
the series estimator and on the linear functional. \cite{Newey1997}
and \cite{Chen2006} provide several examples. In the case of splines with $\mathcal{X}=[0,1]^d$, it has been  established that
$\xi_k \lesssim \sqrt{k}$ and $\sup_{x\in\mathcal{X}}\|\partial_j^mp(x)\| \lesssim k^{1/2 + m}$; see, for example, \cite{Newey1997}. With this basis we have for
\begin{itemize}
\item[1.]  the function $g$ itself: $\theta(x) = g(x)$, $\ell_{\theta}(x)=p(x)$, and $\xi_{k,\theta}\lesssim \sqrt{k}$;

\item[2.]   the derivatives:  $\theta(x) = \partial_j g(x)$, $\ell_\theta(x)= \partial_j p(x)$, $\xi_{k,\theta} \lesssim k^{3/2}$;
\item[3.]  the average derivatives: $\theta = \int \partial_j g(x) d\mu(x)$, $\ell_\theta = \int \partial_j p(x)  d\mu(x)$, and $\xi_{k,\theta} \lesssim 1$,
\end{itemize}
where in the last example it is assumed that ${\rm supp}(\mu) \subset {\rm int}\mathcal{X}$, $x_1$ is continuously distributed with the density bounded below from zero on ${\rm supp}(\mu)$, and  $x \mapsto \partial_l \mu(x)$ is continuous
on ${\rm supp}(\mu)$ with  $|\partial_l \mu(x)| \lesssim 1$ uniformly in $x \in {\rm supp}(\mu)$ for all $l=1,\dots,k$.


We will impose the following regularity condition on the loadings on the coefficients $\ell_\theta(w)$:

\textbf{Condition A.6} (Loadings) \textit{Loadings on the coefficients satisfy (i) $\sup_{w\in\I}1/\|\ell_{\theta}(w)\|\lesssim 1$ and (ii) $\log \xi^L_{k,\theta}\lesssim \log k$.}

The first part of this condition implies that the linear functional is normalized appropriately. The second part is a very mild restriction on the rate of the growth of the Lipschitz coefficient of the map $w \mapsto \theta(w)$.

Under Conditions A.1-A.6, results presented in Lemma \ref{lemma: uniform linearization} on uniform linearization can be extended to cover general linear functionals considered here:
\begin{lemma}[Uniform Linearization for Linear Functionals]\label{lem: uniform linearization functionals}
Assume that Conditions A.1-A.6 are satisfied. 
Then for $\alpha_{\theta}(w)=\ell_{\theta}(w)/\|\ell_{\theta}(w)\|$,
$$
\sqrt{n}\alpha_{\theta}(w)'(\widehat{\beta}-\beta)=\alpha_{\theta}(w)'\mathbb{G}_n[p_i(\epsilon_i+r_i)]+R_{1n}(\alpha_{\theta}(w)),
$$
where $R_{1n}(\alpha_{\theta}(w))$, summarizing the impact of unknown design, obeys
$$
R_{1n}(\alpha(w)) \lesssim_P \sqrt{\frac{\xi_k^2 \log k }{n}} (n^{1/m} \sqrt{\log k}  + \sqrt{k} \cdot \ell_kc_{k})=\bar{R}_{1n}
$$
uniformly over $w \in \mathcal{I}$.
Moreover,
$$
\sqrt{n} \alpha_{\theta}(w)'( \widehat \beta - \beta) = \alpha_{\theta}(w)' \mathbb{G}_n[ p_i \epsilon_i ] + R_{1n}(\alpha_{\theta}(w)) +R_{2n}(\alpha_{\theta}(w)),
$$
where  $R_{2n}(\alpha_{\theta}(w))$, summarizing the impact of approximation error on the sampling error of the estimator, obeys
$$
R_{2n}(\alpha_{\theta}(w)) \lesssim_P \sqrt{\log{k}} \cdot \ell_kc_{k}=\bar{R}_{2n}
$$
uniformly over $w\in\mathcal{I}$.
\end{lemma}

From Lemma \ref{lem: uniform linearization functionals}, we can derive the following theorem on uniform rate of convergence for linear functionals.
\begin{theorem}[Uniform Rate of Convergence for Linear Functionals]\label{theorem: Linearuniform rate}
Assume that Conditions A.1-A.6 are satisfied. Then
\begin{equation}\label{eq: uniform linearization functional 1}
\sup_{w\in\mathcal{I}}\left|\alpha_{\theta}(w)'\mathbb{G}_n[p_i\epsilon_i]\right|\lesssim_P \sqrt{\log k}.
\end{equation}
If, in addition, we assume that (i) $\bar{R}_{1n}+\bar{R}_{2n}\lesssim (\log k)^{1/2}$ and (ii) $\sup_{w\in\I}|r_{\theta}(w)|/\|\ell_{\theta}(w)\| = o((\log k/n)^{1/2})$, then
\begin{equation}\label{eq: uniform linearization functional 2}
\sup_{w\in \I} | \widehat \theta(w) -  \theta(w)| \lesssim_P  \sqrt{\frac{\xi_{k,\theta}^2\log k}{\n}}.
\end{equation}
\end{theorem}

Theorem \ref{theorem: Linearuniform rate} establishes uniform rates that are up to $\sqrt{\log k}$ factor agree with the pointwise rates. The requirement (ii) on the approximation error can be seen as an undersmoothing condition as discussed in Comment \ref{Remark:LinearFunctional}.

Next, we consider the problem of uniform inference for linear functionals based on the series estimator.
We base our inference on the t-statistic process:
\begin{equation}\label{eq: t process definition}
\left \{t(w) = \frac{  \widehat \theta(w) - \theta(w) }{ \widehat \sigma_\theta(w)},  \ \  w \in \I \right\}.
\end{equation}

We present two results for inference on linear functionals. The first result is an extension of Theorem \ref{theorem: uniform normality} on strong approximations to cover the case of linear functionals. As we discussed in Comment \ref{com: on strong approximations}, in order to perform uniform in $w\in\mathcal{I}$ inference on $\theta(w)$, we would like to approximate the distribution of the \textit{whole} process (\ref{eq: t process definition}). However, one can show that this process typically does not have a limit distribution in $\ell^\infty(\mathcal{I})$. Yet, we can construct a Gaussian process that would be close to the process (\ref{eq: t process definition}) for all $w\in\I$ simultaneously with a high probability. Specifically, we will approximate the $t$-statistic process by the following Gaussian coupling:
\begin{equation}
\begin{array}{llll}
\\
& \displaystyle  \Big \{ t_n^*(w) = \frac{  \ell(w)' \Omega^{1/2} \mathcal{N}_k/\sqrt{n} }{  \sigma_\theta(w)},   &  w \in \I \Big \} \\
\end{array}
\end{equation}
where $\mathcal{N}_k$ denotes a vector of $k$ i.i.d. $N(0,1)$ random variables.

\begin{theorem}[Strong Approximation by a Gaussian Process for Linear Functionals]\label{Thm:DistributionsInferential}
Assume that the conditions of Theorem \ref{theorem: uniform normality} and Condition A.6 are satisfied. In addition, assume that (i) $\bar{R}_{2n}\lesssim (\log k)^{1/2}$ and (ii) $\sup_{w\in\I}\sqrt{n}|r_{\theta}(w)|/\|\ell_{\theta}(w)\|=o(a_n^{-1})$. Then
$$
t(w) =_d t^*(w) + o_P(a_n^{-1}) \text{ in } \ell^{\infty}(\I).
$$
\end{theorem}

As in the case of inference on the function $g(x)$, we could also consider the use of the weighted bootstrap method to obtain a result analogous to that in Theorem \ref{Thm:MainBootstrap}. For brevity of the paper, however, we do not consider weighted bootstrap method here.

The second result on inference for linear functionals is new and concerns with the problem of constructing uniform confidence bands for the linear functional $\theta(w)$. Specifically, we are interested in the confidence bands of the form
\begin{equation}\label{eq: confidence bands}
[\dot{\iota}(w),\ddot{\iota}(w)]=\left[\widehat{\theta}(w)-c_n(1-\alpha)\widehat{\sigma}_\theta(w),\widehat{\theta}(w)+c_n(1-\alpha)\widehat{\sigma}_\theta(w)\right], \,\, w\in\mathcal{I}
\end{equation}
where $c_n(1-\alpha)$ is chosen so that $\theta(w)\in[\dot{\iota}(w),\ddot{\iota}(w)]$ for all $w\in\mathcal{I}$ with the prescribed probability $1-\alpha$ where $\alpha\in(0,1)$ is a user-specified level. For this purpose, we would like to set $c_n(1-\alpha)$ as the $(1-\alpha)$-quantile of $\sup_{w\in\mathcal{I}}|t(w)|$.
However, this choice is infeasible because the exact distribution of $\sup_{w\in\mathcal{I}}|t(w)|$ is unknown. Instead, Theorem \ref{Thm:DistributionsInferential} suggests that we can set $c_n(1-\alpha)$ as the $(1-\alpha)$-quantile of $\sup_{w\in\mathcal{I}}|t^*(w)|$ or, if $\Omega$ is unknown and has to be estimated, that we can set
\begin{equation}\label{eq: critical value}
c_n(1-\alpha):=\text{ conditional }(1-\alpha)-\text{quantile of }\sup_{w\in\mathcal{I}}|\widehat{t}^{*}(w)|\text{ given the data}
\end{equation}
where
$$
\widehat{t}_n^*(w):=\frac{\ell(w)^\prime\widehat{\Omega}^{1/2}\mathcal{N}_k/\sqrt{n}}{\widehat{\sigma}_\theta(w)}, \, w\in\I
$$
and $\mathcal{N}_k\sim N(0,I_k)$. Note that $c_n(1-\alpha)$ defined in (\ref{eq: critical value}) can be approximated numerically by simulation. Yet, conditions of Theorem \ref{Thm:DistributionsInferential} are rather strong. Fortunately, \cite{CCK2012} noticed that when we are only interested in the supremum of the process and do not need the process itself, sufficient conditions for the strong approximation can be much weaker. Specifically, we have the following theorem, which is an application of a general result obtained in \cite{CCK2012}:

\begin{theorem}[Strong Approximation of Suprema for Linear Functionals]\label{thm: strong approximation suprema}
Assume that Conditions A.1-A.6 are satisfied with $m\geq 4$. In addition, assume that (i) $\bar{R}_{1n}+\bar{R}_{2n}\lesssim 1/(\log k)^{1/2}$, (ii) $\xi_k\log^2 k/n^{1/2-1/m} \to 0$,  (iii) $1\lesssim \underline{\sigma}^2$, and (iv) $\sup_{w\in\I}\sqrt{n}|r_{\theta}(w)|/\|\ell_{\theta}(w)\|=o(1/(\log k)^{1/2})$. Then
$$
\sup_{w\in\I}|t(w)| =_d \sup_{t\in\I}|t^*(w)| + o_P\left(\frac{1}{\sqrt{\log k}}\right).
$$
\end{theorem}

Construction of uniform confidence bands also critically relies on the following anti-concentration lemma due to
\cite{CCK2012b} (Corollary 2.1):
\begin{lemma}[Anti-concentration for Separable Gaussian Processes]\label{lem: anticoncentration}
Let $Y=(Y_t)_{t\in T}$ be a separable Gaussian process indexed by a semimetric space $T$ such that $E[Y_t]=0$ and $E[Y_t^2]=1$ for all $t\in T$. Assume that $\sup_{t\in T}Y_t<\infty$ a.s. Then $a(|Y|):=E[\sup_{t\in T}|Y_t|]<\infty$ and
$$
\sup_{x\in\mathbb{R}}P\left\{\left|\sup_{t\in T}|Y_t|-x\right|\leq \varepsilon\right\}\leq A\varepsilon a(|Y|)
$$
for all $\varepsilon\geq 0$ and some absolute constant $A$.
\end{lemma}

From Theorem \ref{thm: strong approximation suprema} and Lemma \ref{lem: anticoncentration}, we can now derive the following result on uniform validity of confidence bands in (\ref{eq: confidence bands}):

\begin{theorem}[Uniform Inference for Linear Functionals]\label{thm: confidence bands validity}
Assume that the conditions of Theorem \ref{thm: strong approximation suprema} are satisfied. In addition, assume that $c_n(1-\alpha)$ is defined by (\ref{eq: critical value}). Then
\begin{equation}\label{eq: supremum validity}
P\left\{\sup_{w\in\I}|t_n(w)|\leq c_n(1-\alpha)\right\}=1-\alpha+o(1).
\end{equation}
As a consequence, the confidence bands defined in (\ref{eq: confidence bands}) satisfy
\begin{equation}\label{eq: confidence bands validity}
P \Big\{ \theta(w) \in [\dot{\iota}(w), \ddot{\iota}(w)], \mbox{ for all } w \in \I \Big\} = 1-\alpha + o(1).
\end{equation}
The width of the confidence bands $2c_n(1-\alpha)\widehat{\sigma}_n(w)$ obeys
\begin{equation}\label{eq: confidence bands width}
2c_n(1-\alpha)\widehat{\sigma}_n(w)\lesssim_P \sigma_n(w)\sqrt{\log k}\lesssim \|\ell_{\theta}(w)\|\sqrt{\frac{\log k}{n}}\lesssim \sqrt{\frac{\xi_{k,\theta}^2\log k}{n}}
\end{equation}
uniformly over $w\in\mathcal{I}$.
\end{theorem}
\begin{remark}
(i) This is our fifth (and last) main result in this paper. The theorem shows that the confidence bands constructed above maintain the required level asymptotically and establishes that the
uniform width of the bands is of the same order as the uniform rate
of convergence. Moreover, confidence
intervals are asymptotically similar.

(ii) The proof strategy of Theorem \ref{thm: confidence bands validity} is similar to that proposed in \cite{CLR2008} for inference on the minimum of a function.  Since the limit distribution may not exists, the insight was to use distributions provided by couplings. Because the limit distribution does not necessarily exist, it is not immediately clear
that the confidence bands are asymptotically similar or at least maintain the right asymptotic level.
Nonetheless, we show that the confidence bands are asymptotically similar with the help of anti-concentration lemma stated above.

(iii) Theorem \ref{thm: confidence bands validity} only considers two-sided confidence bands. However, both Theorem \ref{thm: strong approximation suprema} and Lemma \ref{lem: anticoncentration} continue to hold if we replace suprema of absolute values of the processes by suprema of the processes itself, namely if we replace $\sup_{w\in\I}|t_n(w)|$ and $\sup_{w\in\I}|t_n^{*}(w)|$ in Theorem \ref{thm: strong approximation suprema} by $\sup_{w\in\mathcal{I}}t_n(w)$ and $\sup_{w\in\mathcal{I}}t_n^{*}(w)$, respectively, and $\sup_{t\in T}|Y_t|$ in Lemma \ref{lem: anticoncentration} by $\sup_{t\in T}Y_t$. Therefore, we can show that Theorem \ref{thm: confidence bands validity} also applies for one-sided confidence bands, namely Theorem \ref{thm: confidence bands validity} holds with $c_n(1-\alpha)$ defined as the conditional $(1-\alpha)$-quantile of $\sup_{w\in\mathcal{I}}\widehat{t}_n^{*}(w)$ given the data
and the confidence bands defined by
$[\dot{\iota}(w),\ddot{\iota}(w)]:=[\widehat{\theta}(w)-c_n(1-\alpha)\widehat{\sigma}_n(w),+\infty)$ for all $w\in\mathcal{I}$.\qed
\end{remark}

\section{Tools:  Maximal Inequalities for Matrices and Empirical Processes}\label{Sec:Tools}

In this section we collect the main technical tools that our analysis rely upon, namely Khinchin Inequalities for Matrices and Data Dependent Maximal Inequalities.

\subsection{Khinchin Inequalities for Matrices}

For $p\geq 1$, consider the Schatten norm $S_p$ on symmetric $k\times k$ matrices $Q$ defined by
$$
\|Q\|_{S_p} = \left( \sum_{j=1}^k | \lambda_j(Q) |^p \right)^{1/p}
$$
where $\lambda_1(Q),\dots,\lambda_k(Q)$ is the system of eigenvalues of $Q$.
The case $p=\infty$ recovers the operator norm $\|\cdot \|$ and $p=2$ the Frobenius norm.    It is obvious that for any $p\geq 1$
$$
\|Q\| \leq  \|Q\|_{S_p} \leq k^{1/p} \|Q\|.
$$
Therefore, setting $p=\log k$ and observing that $k^{1/\log k}=e$ for any $k\geq 1$, we get the relation:
 \begin{equation}\label{equiv}
\|Q\| \leq \|Q\|_{S_{\log k}} \leq e \|Q\|.
 \end{equation}

\begin{lemma}[Khinchin Inequality for Matrices]\label{lem: Khinchin}
For symmetric $k \times k$-matrices $Q_i$, $i=1,\dots,n$, $2 \leq p < \infty$, and an i.i.d. sequence
of Rademacher variables $\varepsilon_1,\dots,\varepsilon_n$, we have
\begin{equation}\label{eq: khinchin main}
\left \|  \( \En [Q^2_i] \)^{1/2}  \right \|_{S_p} \leq
 \left ( E_{\varepsilon}  \left \|  \Gn [\varepsilon_i Q_i] \right \|^p_{S_p} \right )^{1/p} \leq  C\sqrt{p} \left \|  ( \En[Q^2_i] )^{1/2} \right \|_{S_p}
\end{equation}
for some absolute constant $C$.
As a consequence, we have for $k\geq 2$
\begin{equation}\label{eq: khinchin corollary}
E_{\varepsilon}\left[\left\| \Gn[\varepsilon_i Q_i]\right\|\right] \leq C\sqrt{\log k} \left \| ( \En[ Q_i^2])^{1/2} \right\|
\end{equation}
for some (possibly different) absolute constant $C$.
\end{lemma}
This version of the Khinchin inequality is proven in Section 3 of \cite{Rudelson1999}. We also provide some details of the proof in the Appendix.
The notable feature of this inequality is the $\sqrt{\log k }$ factor instead of the $\sqrt{k}$ factor expected
from the conventional maximal inequalities based on entropy. This inequality due to  \cite{LustPicardPisier} generalizes the Khinchin inequality for vectors. A version of this inequality was derived by \cite{GuedonRudelson} using generalized entropy (majorizing measure) arguments.
This is a striking example where the use of generalized entropy yields drastic improvements over the use of entropy. Prior to this,
\cite{Talagrand1996} provided ellipsoidal examples where the difference between the two approaches was even more extreme.

\subsection{LLN for Matrices}

The following lemma is a variant of a fundamental result obtained by \cite{Rudelson1999}.

\begin{lemma}[Rudelson's LLN for Matrices]\label{rudelson}  Let $Q_1,\dots,Q_n$ be a sequence of independent symmetric non-negative $k \times k$-matrix valued random variables with $k \geq 2$  such that $Q = \En[E[Q_i]]$ and $\|Q_i\| \leq M$ a.s., then for $\widehat Q =  \En[Q_i]$
$$
\Delta : = E\| \widehat Q -  Q  \| \lesssim \frac{M\log k}{n} + \sqrt{ \frac{M \|Q\| \log k}{n}}.
$$
In particular, if $Q_i = p_i p_i'$, with $\|p_i\| \leq \xi_k$ a.s., then
$$
\Delta : = E\| \widehat Q -  Q  \| \lesssim \frac{\xi_k^2\log k}{n} + \sqrt{\frac{\xi_k^2 \|Q\| \log k  }{n}}.
$$
 \end{lemma}
For completeness, we provide the proof of this lemma in the Appendix; see also \cite{Tropp12} for a nice exposition of this result as well as many others concerning with maximal and deviation inequalities for matrices.

\subsection{Maximal Inequalities}
Consider a measurable space $(S, \mathcal{S})$, and a suitably measurable class of functions $\mathcal{F}$  mapping $S$ to $\Bbb{R}$, equipped with a measurable envelope function $F(z) \geq \sup_{f \in \mathcal{F}} |f(z)|$. (By ``suitably measurable"
we mean the condition given in Section 2.3.1 of \cite{vdV-W}; pointwise measurablity and Suslin measurability are sufficient.) The \textit{covering number} $N(\mathcal{F},L^{2}(Q),\varepsilon)$ is the minimal number of $L^{2}(Q)$-balls of radius $\varepsilon$ needed to cover $\mathcal{F}$. The \textit{covering number} relative to the envelope function is given by
\begin{equation}
 N \left( \mathcal{F}%
,L^{2}(Q),\varepsilon \left\Vert F\right\Vert _{Q,2}\right) .
\end{equation}
 The \textit{entropy} is the logarithm of the covering
number.


We rely on the following result.

\begin{proposition}\label{prop1}
Let $(\epsilon_{1},X_{1}), \dots,(\epsilon_{n},X_{n})$ be i.i.d. random vectors, defined on an underlying $n$-fold product probability space,  in $\mathbb{R}^{d+1}$ with $E[\epsilon_{i} | X_{i}] = 0$ and $\sigma^{2} := \sup_{x \in \mathcal{X}} E[ \epsilon_{i}^{2} | X_{i} = x] < \infty$ where $\mathcal{X}$ denotes the support of $X_1$. Let $\mathcal{F}$ be a class of functions on $\mathbb{R}^{d}$ such  that $E [ f(X_{1})^{2} ] = 1$ (normalization) and $\| f \|_{\infty} \leq b$ for all $f \in \mathcal{F}$.
Let $\mathcal{G} := \{ \mathbb{R}\times \mathbb{R} \ni (\epsilon,x)  \mapsto \epsilon f(x) : f \in \mathcal{F} \}$. Suppose that there exist constants $A > e^{2}$ and $V \geq 2$ such that
$$\sup_{Q} N(\mathcal{G}, L^{2}(Q), \varepsilon \| G \|_{L^{2}(Q)}) \leq (A/\varepsilon)^{V}$$ for all $0 < \varepsilon \leq 1$ for the envelope $G(\epsilon,x) := | \epsilon| b$.
If for some $m > 2$ $E[ | \epsilon_{1} |^{m} ] < \infty$, then
\begin{equation*}
E \left [ \left \|  \sum_{i=1}^{n} \epsilon_{i} f(X_{i}) \right \|_{\mathcal{F}} \right] \leq C \left [ (\sigma + \sqrt{E[ | \epsilon_{1} |^{m} ]})  \sqrt{ n V \log (Ab)} + V b^{m/(m-2)} \log (Ab) \right ],
\end{equation*}
where $C$ is a universal constant.
\end{proposition}

The proof is based on a truncation argument and maximal inequalities for uniformly bounded classes of functions developed in \cite{GK06}. We recall its version.

\begin{theorem}[\cite{GK06}]
\label{GK06}
Let $\xi_{1},\dots,\xi_{n}$ be i.i.d. random variables taking values in a measurable space $(S,\mathcal{S})$ with common distribution $P$, defined on the underlying $n$-fold product probability space.
Let $\mathcal{F}$ be a suitably measurable class of functions mapping $S$ to $\Bbb{R}$ with a measurable envelope $F$. Let $\sigma^{2}$ be a constant such that $
\sup_{f \in \mathcal{F}} \var(f) \leq \sigma^{2} \leq \| F \|_{L^{2}(P)}^{2}$. Suppose that there exist constants $A > e^{2}$ and $V \geq 2$ such that
$\sup_{Q} N(\mathcal{F}, L^{2}(Q), \varepsilon \| F \|_{L^{2}(Q)}) \leq (A/\varepsilon)^{V}$ for all $0 < \varepsilon \leq 1$. Then,
\begin{equation*}
E \left [ \left \| \sum_{i=1}^{n} \{ f(\xi_{i}) - E[ f(\xi_{1}) ]\} \right \|_{\mathcal{F}} \right] \leq C \left [ \sqrt{n \sigma^{2} V \log \frac{A \| F \|_{L^{2}(P)}}{\sigma}} + V \| F \|_{\infty} \log \frac{A \| F \|_{L^{2}(P)}}{\sigma} \right ],
\end{equation*}
where $C$ is a universal constant.
\end{theorem}

\section*{ Acknowledgements.} This paper was presented and first circulated in a series of lectures
given by Victor Chernozhukov at ``Stats in the Ch\^{a}teau" Statistics Summer School on ``Inverse Problems and High-Dimensional Statistics" in 2009 near Paris.  Participants, especially Xiaohong Chen, and one of several referees made numerous helpful suggestions. We also thank Bruce Hansen for extremely useful comments.

\appendix

\section{Proofs}

\subsection{Proofs of Sections \ref{Sec:Setup} and \ref{Sec:Approx}}

\begin{proof}[Proof of Proposition \ref{lemma:StabilityValues}]
Recall that $p(x)=(p_1(x),\dots,p_k(x))'$. Since $dF/d\mu$ is bounded above and away from zero on $\mathcal{X}$, and regressors $p_1(x),\dots,p_k(x)$ are orthonormal under $(\mathcal{X},\mu)$, we have
\begin{align*}
\|\gamma\|^2&=\int_{\mathcal{X}}(\gamma'p(x))^2d\mu(x)\lesssim \int_{\mathcal{X}}(\gamma'p(x))^2(dF/d\mu)(x)d\mu(x)\\
&=\int_{\mathcal{X}}(\gamma'p(x))^2dF(x)\lesssim \int_{\mathcal{X}}(\gamma'p(x))^2d\mu(x)=\|\gamma\|^2
\end{align*}
uniformly over all $\gamma\in S^{k-1}$. The asserted claim follows.
\end{proof}

\begin{proof}[Proof of Proposition \ref{lem: simple bound on l}]
Fix $f\in\G$. Let
$$
\beta_f^\star:=\arg\min_b\|f-p'b\|_{F,\infty}.
$$
Then
$$
\|r_f\|_{F,\infty}=\|f-p'\beta_f\|_{F,\infty}\leq \|f-p'\beta_f^\star\|_{F,\infty}+\|p'\beta_f^\star-p'\beta_f\|_{F,\infty}\leq c_k+\|p'\beta_f^\star-p'\beta_f\|_{F,\infty}.
$$
Further, first order conditions imply that $\beta_f=Q^{-1}E[p(x_1)f(x_1)]$, and so for any $x\in\mathcal{X}$,
\begin{align*}
p(x)'\beta_f^\star-p(x)'\beta_f&=p(x)'Q^{-1}Q\beta_f^\star-p(x)'Q^{-1}E[p(x_1)f(x_1)]\\
&=p(x)'Q^{-1}E[p(x_1)(p(x_1)'\beta_f^\star-f(x_1))].
\end{align*}
This implies that
$$
\|p'\beta_f^\star-p'\beta_f\|_{F,\infty}\leq \xi_k\|E[p(x_1)(p(x_1)'\beta_f^\star-f(x_1))]\|.
$$
Moreover, since $E[p(x_1)p(x_1)']=Q=I$, $E[p_j(x_1)(p(x_1)'\beta_f^\star-f(x_1))]$ is the coefficient on $p_j(x_1)$ of the projection of $p(x_1)'\beta_f^\star-f(x_1)$ onto $p(x_1)$, and so
$$
\|E[p(x_1)(p(x_1)'\beta_f^\star-f(x_1))]\|\leq \left(E[(p(x_1)'\beta_f^\star-f(x_1))^2]\right)^{1/2}\leq c_k.
$$
Conclude that
$$
\|r_f\|_{F,\infty}\leq c_k+\xi_kc_k=c_k(1+\xi_k),
$$
and so Condition A.3 holds with $\ell_k=1+\xi_k$.
This completes the proof of the proposition.
\end{proof}

\begin{proof}[Proof of Proposition \ref{lem: approximation theory bound}]
Fix $f\in\G$. Define $\beta_f^\star$ by
$$
\beta_f^\star:=\arg\min_b\|f-p'b\|_{F,\infty}.
$$
Note that for any functions $f_1,f_2\in\bar{\mathcal{G}}$, $\beta_{f_1+f_2}=\beta_{f_1}+\beta_{f_2}$. Therefore,
$$
\beta_{f-p'\beta_f^\star}=\beta_f-\beta_f^\star,
$$
and so we obtain
\begin{align*}
\|r_f\|_{F,\infty}&\leq \|f-p'\beta_f^\star\|_{F,\infty}+\|p'\beta_f^\star-p'\beta_f\|_{F,\infty}\\
& = \|f-p'\beta_f^\star\|_{F,\infty}+\|p'\beta_{f-p'\beta_f^\star}\|_{F,\infty} \\
&\leq\|f-p'\beta_f^\star\|_{F,\infty}+\widetilde{\ell}_k\|f-p'\beta_f^\star\|_{F,\infty}
\leq (1+\widetilde{\ell}_k)\inf_b\|f-p'b\|_{F,\infty}
\end{align*}
where on the third line we used the definition of $\widetilde{\ell}_k$. Hence,
$$
\|r_f\|_{F,\infty}\leq (1+\widetilde{\ell}_k)\inf_b\|f-p'b\|_{F,\infty}.
$$
Next,
$$
c_k\geq \sup_{f\in\mathcal{G}}\inf_b\|f-p'b\|_{F,\infty}
$$
implies that
$$
\|r_f\|_{F,\infty}\leq c_k(1+\widetilde{\ell}_k),
$$
and so Condition A.3 holds with $\ell_k=1+\widetilde{\ell}_k$. This completes the proof of the proposition.
\end{proof}

\subsection{Proofs of Section \ref{Sec:L2Rate}}

\begin{proof}[Proof of Theorem \ref{theorem: L2 rate}] We have that
$$
\|\widehat g - g\|_{F,2} \leq  \| p'\widehat \beta - p'\beta\|_{F,2} + \| p'\beta - g\|_{F,2} \leq  \|p'\widehat \beta - p'\beta\|_{F,2} + c_k
$$
where under the normalization $Q=E[p(x_i)p(x_i)']=I$ we have
$$
\|p'\widehat \beta - p'\beta\|_{F,2} = \left[\int  (\widehat \beta - \beta)' p(x) p(x)' (\widehat \beta - \beta)  dF(x) \right]^{1/2} = \| \widehat \beta - \beta \|.
$$
Further,
$$
\| \widehat \beta - \beta \| = \|  \widehat Q^{-1} \En [ p_i (\epsilon_i + r_i)]\| \leq \|  \widehat Q^{-1} \En [ p_i \epsilon_i] \| +  \|\widehat Q^{-1} \En [ p_i r_i] \|.
$$
 By the Matrix LLN (Lemma \ref{rudelson}), which is the critical step, we have that
$$
\|\widehat Q - Q\| \to_P 0  \text{\,\,  if \,} \frac{\xi_k^2 \log k}{n} \to 0.
$$
Therefore, wp $\to $ 1, all eigenvalues of $\widehat{Q}$ are bounded away from zero. Indeed, if at least one eigenvalue of $\widehat{Q}$  is strictly smaller than 1/2, then there exists a vector $a\in S^{k-1}$ such that $a'\widehat{Q}a<1/2$, and so
$$
\|\widehat{Q}-Q\|\geq |a'(\widehat{Q}-Q)a|=|a'\widehat{Q}a-a'a|=|a'\widehat{Q}a-1|>1/2.
$$
Hence, wp $\to $ 1, all eigenvalues of $\widehat{Q}$ are not smaller than 1/2. Therefore,
$$
\|  \widehat Q^{-1} \En [ p_i \epsilon_i] \| \lesssim_P \| \En [ p_i \epsilon_i] \| \lesssim_P \sqrt{k/n}
$$
where the second inequality follows from
$$
E\left[ \| \En [ p_i \epsilon_i] \|^2 \right] =  E[\epsilon_i^2 p_i'p_i/n ] =  E[\sigma^2_i p_i'p_i/n ] \lesssim E[p_i'p_i/n ] = k/n
$$
since $\sigma^2_i\leq \bar{\sigma}^2$ is bounded.
Moreover, since
$\widehat r_i: = p_i' \widehat Q^{-1} \En [p_i r_i]$ is a sample projection of $r_i$ on $p_i$,
\begin{equation}\label{eq: derivation bound 1}
\|\widehat Q^{-1/2} \En [ p_i r_i] \|^2 = \En [r_i\widehat r_i] = \En [\widehat r_i^2] \leq \En[r_i^2] \lesssim_P  E[r_i^2] \leq c_k^2,
\end{equation}
by Markov's inequality. Therefore, when $c_k \to 0$,
$$
 \|\widehat Q^{-1} \En \left[ p_i r_i\right] \| \lesssim_P  \|\widehat Q^{-1/2} \En [ p_i r_i] \| \lesssim_P c_k
$$
where the first inequality follows from all eigenvalues of $\widehat{Q}^{1/2}$ being bounded away from zero wp $\to $ 1 and the second from (\ref{eq: derivation bound 1}). This completes the proof of  (\ref{eq: l2 rate 1}).

Further, note that
\begin{align}
E\left[\|\En[p_i r_i]\|^2\right]
&=\frac{1}{n^2}E\left[\sum_{j=1}^k\Big(\sum_{i=1}^n p_j(x_i) r(x_i)\Big)^2\right]\nonumber\\
&=\frac{1}{n}E\left[\sum_{j=1}^k p_j(x_1)^2r(x_1)^2\right]\leq \left(\frac{\ell_k c_k}{\sqrt{n}}\right)^2E[\|p(x_1)\|^2]=\left(\frac{\ell_kc_k\sqrt{k}}{\sqrt{n}}\right)^2\label{eq: alternative bounds}
\end{align}
where we used $E[ p_{i} r_{i} ] = 0$. Alternatively, the first term in (\ref{eq: alternative bounds}) can be bounded from above as
$$
\frac{1}{n}E\left[\sum_{j=1}^k p_j(x_1)^2r(x_1)^2\right]\leq \frac{1}{n}E\left[ \xi_k^2 r(x_1)^2\right]\leq \frac{\xi_k^2 c_k^2}{n}.
$$
Therefore, when $c_k \not \to 0$,
$$
\|\widehat Q^{-1}\En[p_i r_i]\|  \leq  \|\widehat Q^{-1} \|  \|\En[p_i r_i]\| \lesssim_P (\ell_k c_k\sqrt{ k/n})\wedge(\xi_kc_k/\sqrt{n}),
$$
and so (\ref{eq: l2 rate 2}) follows. This completes the proof of the theorem.
\end{proof}

\subsection{Proofs of Section \ref{Sec:Pointwise}}

\begin{proof}[Proof of Lemma \ref{lemma:linearization}]
Decompose
$$
\sqrt{n}\alpha'(\widehat \beta - \beta) = \alpha' \Gn[p_i(\epsilon_i + r_i)] + \alpha'[\widehat Q^{-1} - I ]\Gn[ p_i(\epsilon_i + r_i)].
$$

We divide the proof in three steps. Steps 1 and 2 establish (\ref{eq: lin2}), the bound on $R_{1n}(\alpha)$. Step 3 proves (\ref{eq: lin4}), the bound on $R_{2n}(\alpha)$.

\textbf{Step 1.} Conditional on $X=[x_1,\ldots,x_n]$, the term
$$
\alpha'[\widehat Q^{-1} - I ]\Gn[ p_i \epsilon_i]
$$
has mean zero and variance bounded by $\bar \sigma^2\alpha'[\widehat Q^{-1} - I ] \widehat Q  [\widehat Q^{-1} - I] \alpha$. Next, as in the proof of Theorem \ref{theorem: L2 rate}, wp $\to $ 1, all eigenvalues of $\widehat{Q}$ are bounded away from zero and from above, and so
$$
\bar \sigma^2\alpha'[\widehat Q^{-1} - I ] \widehat Q  [\widehat Q^{-1} - I] \alpha \lesssim  \bar \sigma^2\|\widehat Q\| \|\widehat Q^{-1} \|^2 \| \widehat Q - I \|^2 \lesssim_P  \frac{\xi^2_k \log k}{n}
$$
where the second inequality follows from Matrix LLN (Lemma \ref{rudelson}) and $\bar{\sigma}^2\lesssim 1$.
We then conclude by Chebyshev's inequality that
$$
\alpha'[\widehat Q^{-1} - I ]\Gn[ p_i \epsilon_i]
\lesssim_P \sqrt{\frac{\xi^2_k \log k}{n}}.$$

\textbf{Step 2.} By Matrix LLN (Lemma \ref{rudelson}), $\|\widehat Q - I\|\lesssim_P(\xi_k^2\log k / n )^{1/2}$, and so
\begin{align*}
|\alpha'(\widehat Q^{-1} - I)\Gn[p_i r_i]| &\leq    \|\widehat Q^{-1} - I \| \cdot \|\Gn[p_i r_i]\| \\
&\leq   \| \widehat Q^{-1}\|\cdot \|\widehat Q - I \|\cdot  \|\Gn[p_i r_i]\|
\lesssim_P   \sqrt{\frac{\xi_k^2 \log k}{ n}} \ell_kc_k\sqrt{k},
\end{align*}
where we used the bound $\|\Gn[p_i r_i]\|\lesssim_P \ell_kc_k\sqrt{k}$ obtained in the proof of Theorem \ref{theorem: L2 rate}.
Steps 1 and 2 give the linearization result (\ref{eq: lin2}).

\textbf{Step 3.} Since $E[p_i r_i] =0$, the term
$$
R_{2n}(\alpha) = \alpha' \Gn[p_i r_i]
$$
has mean zero and variance
$$
E[(\alpha'p_i r_i)^2] \leq E[(\alpha'p_i)^2] \ell_k^2 c^2_k \leq \ell^2_k c^2_k.
$$
Thus, (\ref{eq: lin4}) follows from Chebyshev's inequality. This completes the proof of the lemma.
\end{proof}

\begin{proof}[Proof of Theorem \ref{theorem: pointwise}]
Note that (\ref{eq: pointwise normality 2}) follows by applying (\ref{eq: pointwise normality 1}) with $\alpha=p(x)/\|p(x)\|$, and (\ref{eq: pointwise normality 3}) follows directly from (\ref{eq: pointwise normality 2}). Therefore, it suffices to prove (\ref{eq: pointwise normality 1}).

Observe that for any $\alpha\in S^{k-1}$, $1\lesssim \|\alpha'\Omega^{1/2}\|$ because $1\lesssim \underline{\sigma}^2\leq \sigma_i^2$ and
\begin{equation}\label{omega bounded}
\Omega \geq \Omega_0 \geq \underline{\sigma}^2 Q^{-1}
\end{equation}
in the positive semidefinite sense.
Further, by condition (iii) of the theorem and Lemma \ref{lemma:linearization}, $R_{1n}(\alpha)=o_P(1)$ (note that we can apply Lemma \ref{lemma:linearization} because $\bar{\sigma}^2\lesssim 1$ follows from condition (i) and $\xi_k^2\log k/n\to 0$ follows from condition (iii) of the theorem).
Therefore, we can write
$$
\frac{\sqrt{n}\alpha'}{\|\alpha' \Omega^{1/2}\|} (\widehat \beta - \beta) =  \frac{\alpha'}{\|\alpha' \Omega^{1/2}\|}  \Gn[p_i (\epsilon_i + r_i)] + o_P(1)=  \sum_{i=1}^n \omega_{ni} (\epsilon_i + r_i)+o_P(1),
$$
where
$$
\omega_{ni} =\frac{\alpha'}{\|\alpha' \Omega^{1/2}\|} \frac{p_i}{\sqrt{n}}, \  \  |\omega_{ni}| \lesssim \frac{\xi_k}{\sqrt{n}},  \ \   |\epsilon_i + r_i| \leq |\epsilon_i| + \ell_k c_k.
$$
Further, it follows from (\ref{omega bounded}) that
\begin{equation}\label{omegaSigma}
n E |\omega_{ni}|^2  \leq  E[(\alpha'p_i)^2]/(\alpha'\Omega\alpha) \leq 1/\underline{\sigma}^2\lesssim 1.
\end{equation}
Now we verify Lindberg's condition for the CLT. First, by construction we have
$$
\text{ var }  \Big (\sum_{i=1}^n \omega_{ni} (\epsilon_i + r_i) \Big ) = 1.
$$
Second, for each $\delta>0$
$$
\sum_{i=1}^n E \[|\omega_{ni}|^2 (\epsilon_i + r_i)^2 1 \{  |\omega_{ni} (\epsilon_i + r_i)| > \delta \} \]\to 0,
$$
since  the left hand side is bounded by
 $$2n E\[ |\omega_{ni}|^2 \epsilon_i^21 \{ |\epsilon_i| + \ell_kc_k > \delta/|\omega_{ni}| \}\] \ \ + \ \   2 n E \[|\omega_{ni}|^2  \ell_k^2 c_k^2 1 \{ |\epsilon_i| + \ell_kc_k > \delta/|\omega_{ni}| \}\],$$
and both terms go to zero. Indeed, the first term is bounded from above for some $c>0$ by
 \begin{eqnarray*}
& & 2n E \[|\omega_{ni}|^2 E\[\epsilon_i^2 1\{ |\epsilon_i| + \ell_k c_k> c\delta \sqrt{n}/\xi_k\} |x_i\]\]  \\
& &  \quad \lesssim  n E \[|\omega_{ni}|^2 \] \cdot \sup_{x \in \X} E\[\epsilon_i^2 1\{ |\epsilon_i| + \ell_k c_k> c\delta \sqrt{n}/\xi_k\} |x_i =x\] =o(1)
 \end{eqnarray*}
where we used (\ref{omegaSigma}), the uniform integrability in the condition (i) and $c\delta \sqrt{n}/\xi_k-  \ell_k c_k \to \infty$, which follows from the condition (iii); the second term is bounded from above by
\begin{eqnarray*}
& &
2n E \[|\omega_{ni}|^2 \ell^2_k c_k^2 P\[ |\epsilon_i| +  \ell_k c_k > c\delta \sqrt{n}/\xi_k |x_i\]\]  \\
& & \quad \lesssim n E \[|\omega_{ni}|^2 \ell^2_k c_k^2\] \cdot  \sup_{x \in \X} P\[ |\epsilon_i| + \ell_k c_k> c\delta \sqrt{n}/\xi_k |x_i=x\]  \\
& & \quad \lesssim \ell^2_k c_k^2 \cdot \frac{\bar \sigma^2}{[ c\delta \sqrt{n}/\xi_k -  \ell_k c_k]^2}  = o(1)
\end{eqnarray*}
by Chebyshev's inequality where we used (\ref{omegaSigma}), $c\delta \sqrt{n}/\xi_k- \ell_k c_k \to \infty$, and  $\ell_k c_k =o( \delta \sqrt{n}/\xi_k)$. \end{proof}

\subsection{Proofs of Section \ref{Sec:Uniform}}

\begin{proof}[Proof of Lemma \ref{lemma: uniform linearization}]  Decompose
$$
\sqrt{n}\alpha(x)'(\widehat \beta - \beta) = \alpha(x)' \mathbb{G}_{n}[p_i(\epsilon_i + r_i)] + \alpha(x)'[\widehat Q^{-1} - I ]\mathbb{G}_{n}[ p_i(\epsilon_i + r_i)].
$$
We divide the proof in three steps. Steps 1 and 2 establish (\ref{eq: lin2U}), the bound on $R_{1n}(\alpha(x))$. Step 3 proves (\ref{eq: lin4U}), the bound on $R_{2n}(\alpha(x))$.

\textbf{Step 1.} Here we show that
\begin{equation}\label{eq: A.26}
\sup_{x \in \mathcal{X}} \left| \alpha(x)'[\widehat Q^{-1} - I ]\mathbb{G}_{n}[  p_i \epsilon_{i} ] \right| \lesssim_{P} n^{1/m} \sqrt{\frac{\xi_{k}^{2} \log^{2} k}{n}}.
\end{equation}

Conditional on the data, let $T := \{ t = (t_1,\dots,t_{n}) \in \mathbb{R}^{n} : t_{i} = \alpha(x)'(\widehat Q^{-1} - I)p_i\epsilon_i, x \in \mathcal{X} \}$. Define the norm $\| \cdot \|_{n,2}$ on $\mathbb{R}^{n}$ by $\| t \|_{n,2}^{2} = n^{-1} \sum_{i=1}^{n} t_{i}^{2}$. Recall that for $\varepsilon > 0$, an $\varepsilon$-net of a normed space $(T,\|\cdot\|_{n,2})$ is a subset $T_{\varepsilon}$ of $T$ such that for every $t \in T$ there exists a point $t_{\varepsilon} \in T_{\varepsilon}$ with $\|t-t_{\varepsilon}\|_{n,2} < \varepsilon$.  The covering number $N(T,\|\cdot\|_{n,2},\varepsilon)$ of $T$ is the infimum of the cardinality of $\varepsilon$-nets of $T$.

Let $\eta_{1},\dots,\eta_{n}$ be independent Rademacher random variables ($P(\eta_1=1)=P(\eta_1=-1)=1/2$) that are independent of the data, and denote $\eta=(\eta_1,\dots,\eta_n)$. Also, let $E_\eta[\cdot]$ denote the expectation with respect to the distribution of $\eta$. Then by Dudley's inequality \cite[]{Dudley1967},
\begin{equation*}
E_{\eta} \left[ \sup_{x \in \mathcal{X}} | \alpha(x)'[\widehat Q^{-1} - I ]\mathbb{G}_{n}[ \eta_{i}  p_i \epsilon_{i} ] |\right] \lesssim \int_{0}^{\theta} \sqrt{\log N( T, \| \cdot \|_{n,2}, \varepsilon )} d \varepsilon,
\end{equation*}
where
$$
\theta := 2 \sup_{t  \in T} \| t \|_{n,2} = 2 \sup_{x \in \mathcal{X}} \left(\En[(\alpha(x)'(\widehat Q^{-1} - I)p_i\epsilon_i )^2]\right)^{1/2}\leq 2 \max_{1 \leq i \leq n} | \epsilon_{i} | \| \widehat{Q}^{-1} - I \| \| \widehat{Q} \|^{1/2}.
$$
Since for any $x, \widetilde x \in \mathcal{X}$,
\begin{align*}
&\left( \En[(\alpha(x)'(\widehat Q^{-1} - I)p_i\epsilon_i - \alpha(\widetilde x)'(\widehat Q^{-1} - I)p_i\epsilon_i)^2] \right)^{1/2}\\
&\qquad \leq \max_{1\leq i\leq n}|\epsilon_i| \| \alpha(x) - \alpha(\widetilde x) \|  \| \widehat{Q}^{-1}-I \|\|\widehat{Q}\|^{1/2} \\
&\qquad \leq \xi^L_k  \max_{1 \leq i \leq n} | \epsilon_{i} | \| \widehat{Q}^{-1}-I \|\|\widehat{Q}\|^{1/2} \| x - \widetilde x \|,
\end{align*}
we have for some $C>0$,
\begin{equation*}
 N( T, \| \cdot \|_{n,2}, \varepsilon ) \leq \left ( \frac{C \xi^L_k \max_{1 \leq i \leq n} | \epsilon_{i} |  | \| \widehat{Q}^{-1}-I \|\|\widehat{Q}\|^{1/2}}{\varepsilon} \right )^{d}.
\end{equation*}
Thus we have
\begin{equation*}
 \int_{0}^{\theta} \sqrt{\log N( T, \| \cdot \|_{n,2}, \varepsilon )} d \varepsilon \leq \max_{1 \leq i \leq n} | \epsilon_{i} | \| \widehat{Q}^{-1}-I \| \|\widehat{Q}\|^{1/2}\int_{0}^{2 } \sqrt{ d \log (C \xi^L_k/\varepsilon}) d \varepsilon.
\end{equation*}
By A.4, we have $E[\max_{1 \leq i \leq n} | \epsilon_{i} | \mid X] \lesssim_{P} n^{1/m}$ where $X=(x_1,\dots,x_n)$. In addition, note that $\xi_k^{2m/(m-2)}\log k/n\lesssim 1$ for $m>2$ implies that $\xi_k^2\log k/n\to 0$. Therefore, we have $\| \widehat{Q}^{-1}-I \| \lesssim_{P} ( \xi_{k}^{2} \log k/n)^{1/2}$ and $\| \widehat{Q} \| \lesssim_{P} 1$. Hence, it follows from $\log \xi^L_k \lesssim  \log k$ that
\begin{align*}
E\left[ \sup_{x \in \mathcal{X}} | \alpha(x)'[\widehat Q^{-1} - I ]\mathbb{G}_{n}[  p_i \epsilon_{i} ] | \mid X\right] &\leq 2 E\left[ E_{\eta} [ \sup_{x \in \mathcal{X}} | \alpha(x)'[\widehat Q^{-1} - I ]\mathbb{G}_{n}[ \eta_{i}  p_i \epsilon_{i} ] | ] \mid X\right] \\
&\lesssim_{P} n^{1/m} \sqrt{\frac{\xi_{k}^{2} \log^{2} k}{n}},
\end{align*}
where the first line is due to the symmetrization inequality. Thus, (\ref{eq: A.26}) follows.

\textbf{Step 2.} Observe that
\begin{equation*}
\sup_{x \in \mathcal{X}}|\alpha(x)'(\widehat Q^{-1} - I)\mathbb{G}_{n} [p_i r_i]|  \leq   \| \widehat Q^{-1} - I \|\cdot \|\mathbb{G}_{n} [p_i r_i]\|\lesssim_P \sqrt{\frac{\xi_k^2\log k}{n}}\ell_kc_k\sqrt{k}
\end{equation*}
where the second inequality was shown in the proof of Lemma \ref{lemma:linearization}.
Now, Steps 1 and 2 give the linearizarion result (\ref{eq: lin2U}).

\textbf{Step 3.} We wish to bound $\sup_{x \in \mathcal{X}} |\alpha(x)' \mathbb{G}_{n} [ p_{i} r_{i} ]|$. We use Theorem \ref{GK06}.
Consider the class of functions
\begin{equation*}
\mathcal{F} := \{ \alpha(x)'p(\cdot) r(\cdot) : x \in \mathcal{X} \}.
\end{equation*}
Then, $| \alpha(x)'p(\cdot) r(\cdot) | \leq \ell_kc_{k} \xi_{k}$, $E[(\alpha(x)'p(x_i)r(x_i))^2]\leq (\ell_kc_k)^2$, and for any $x, \widetilde x \in \mathcal{X}$,
\begin{equation*}
| \alpha(x)'p(\cdot) r(\cdot) - \alpha(\widetilde x)'p(\cdot) r(\cdot) | \leq \ell_kc_{k} \xi^L_k \xi_{k} \| x - \widetilde x \|,
\end{equation*}
so that for some $C>0$,
\begin{equation*}
\sup_{Q} N(\mathcal{F}, L^{2}(Q), \varepsilon \ell_kc_{k} \xi_{k} ) \leq \left ( \frac{C  \xi^L_k}{\varepsilon} \right)^{d}.
\end{equation*}
Thus, using conditions (ii) and (iii) of A.5, we have by Theorem \ref{GK06} that
\begin{equation*}
E\left[ \sup_{x \in \mathcal{X}} |\alpha(x)' \mathbb{G}_{n} [ p_{i} r_{i} ]| \right] \lesssim \ell_kc_{k} \sqrt{\log k} + \ell_kc_{k} \frac{\xi_{k} \log k}{\sqrt{n}} \lesssim \ell_kc_{k} \sqrt{\log k},
\end{equation*}
where we have used the fact that
\begin{equation*}
\frac{\xi_{k} \log k}{\sqrt{n}} = \sqrt{\log k} \sqrt{\frac{\xi_{k}^{2} \log k}{n}} = o(\sqrt{\log k}).
\end{equation*}
Therefore, we have by Markov's inequality
\begin{equation}\label{Eq:boundpr}
\sup_{x \in \mathcal{X}} |\alpha(x)' \mathbb{G}_{n} [ p_{i} r_{i} ]| \lesssim_{P}  \ell_kc_{k}\sqrt{\log k}.
\end{equation}
So, the linearization result (\ref{eq: lin4U}) follows.
This completes the proof.
\end{proof}

\begin{proof}[Proof of Theorem \ref{theorem: uniform rate}]
Note that (\ref{eq: uniform rate 2}) and (\ref{eq: uniform rate 3}) follow from (\ref{eq: 4.13}) and Lemma \ref{lemma: uniform linearization}. Therefore, it suffices to prove (\ref{eq: 4.13}), and so we wish to bound $\sup_{x \in \mathcal{X}}| \alpha(x)' \mathbb{G}_{n}[p_{i} \epsilon_{i}]|$. To this end, we use Proposition \ref{prop1}. Consider the class of functions
\begin{equation*}
\mathcal{G} := \{ (\epsilon,x) \mapsto \epsilon \alpha(v)'p(x) : v \in \mathcal{X} \}.
\end{equation*}
Then, $| \alpha(v)'p(x_{i}) | \leq  \xi_{k}$, $\var(\alpha(v)'p(x_{i})) = 1$ and for any $v, \widetilde v \in \mathcal{X}$,
\begin{equation*}
| \epsilon \alpha(v)'p(x) - \epsilon \alpha(\widetilde v)'p(x) | \leq | \epsilon | \xi^L_k \xi_{k} \| v - \widetilde v \|.
\end{equation*}
Thus, taking $G(\epsilon,x) := | \epsilon | \xi_{k}$, we have
\begin{equation*}
\sup_{Q} N(\mathcal{G},L^{2}(Q),\varepsilon \| G \|_{L^{2}(Q)}) \leq \left ( \frac{C \xi^L_k}{\varepsilon} \right)^{d}.
\end{equation*}
Therefore, by Proposition \ref{prop1}, we have
\begin{equation}\label{Eq:bounder}
E\left[ \sup_{x \in \mathcal{X}}| \alpha(x)' \mathbb{G}_{n}[p_{i} \epsilon_{i}]| \right] \lesssim \sqrt{\log k} +  \frac{ \xi_{k}^{m/(m-2)} \log k}{\sqrt{n}} \lesssim \sqrt{\log k},
\end{equation}
where we have used the following inequality
\begin{equation*}
\frac{ \xi_{k}^{m/(m-2)} \log k}{\sqrt{n}} = \sqrt{\log k} \cdot \sqrt{\frac{\xi_{k}^{2m/(m-2)} \log k}{n}} \lesssim \sqrt{\log k}.
\end{equation*}
This completes the proof.
\end{proof}

\begin{proof}[Proof of Theorem \ref{theorem: uniform normality}] The proof follows similarly to that in  \cite{CLR2008}.
We shall apply Yurinskii's
coupling (see Theorem 10 in \cite{Pollard2002}):

Let $\zeta_1,\dots,\zeta_n$ be independent $k$-vectors with
$E[\zeta_i]= 0$ for each $i$, and $\Delta := \sum_{i=1}^n E \|\zeta_i\|^3$ finite. Let $S$ denote
denote a copy of $\zeta_1 +\cdots + \zeta_n$ on a sufficiently rich probability space $(\Omega, \mathcal{A}, P)$. For each $\delta>0$ there exists a random vector $T$ in this space with a $N(0, \text{var}(S))$ distribution such that
$$
P\{ \| S- T\| > 3 \delta \} \leq C_0 B \left( 1  + \frac{|\log (1/B)|}{k}   \right) \text{ where }
B:= \Delta k \delta^{-3},
$$
for some universal constant $C_0$.

In order to apply the coupling, consider  a copy of the first order approximation to our estimator on a suitably rich probability space
$$
\frac{1}{\sqrt{n}} \sum_{i=1}^n \zeta_i, \ \ \zeta_i = \Omega^{-1/2} p_i (\epsilon_i + r_i).
$$
When $\bar{R}_{2n}= o_P(a_n^{-1})$, a similar argument can be used with $\zeta_i = \Omega^{-1/2} p_i (\epsilon_i + r_i)$ replaced by $\zeta_i = \Omega^{-1/2} p_i \epsilon_i$.
As in the proof of Theorem \ref{theorem: pointwise}, all eigenvalues of $\Omega$ are bounded away from zero. Therefore,
\begin{align*}
E\|\zeta_i\|^3  & \lesssim  E[ \| p_i (\epsilon_i + r_i)\|^3] \\
& \lesssim    E [\| p_i\|^3 (|\epsilon_i|^3 + |r_i|^3)]  \\
& \lesssim   E[\|p_i\|^3] (1 + \ell_k^3 c_k^3) \\
& \lesssim   E[\|p_i\|^2 ]\xi_k  (1 + \ell_k^3 c_k^3) \\
& \lesssim   k \xi_k  (1 + \ell_k^3 c_k^3)
\end{align*}
where we used the assumption that $\sup_{x\in\mathcal{X}}E[|\epsilon_i|^3|x_i=x]\lesssim 1$.
Therefore, by Yurinskii's coupling, for each $\delta>0$,
\begin{align*}
 P \left \{ \left\| \frac{\sum_{i=1}^n \zeta_i}{\sqrt{n}} - \mathcal{N}_k \right\|  \geq 3 \delta a_n^{-1} \right\}
&  \lesssim  \frac{ n k^2 \xi_k  (1 + \ell_k^3 c_k^3)}{(\delta a_n^{-1} \sqrt{n} )^3} \left( 1 + \frac{\log( k^2 \xi_k  (1 + \ell_k^3 c_k^3))}{k}\right) \\
&  \lesssim \frac{ a_n^3  k^2 \xi_k  (1 + \ell_k^3 c_k^3)  }{\delta^3 n^{1/2}} \left( 1 + \frac{\log n}{k} \right) \to 0
\end{align*}
because
$a_n^6  k^4 \xi^2_k  (1 + \ell_k^3 c_k^3)^2 \log^2 n/n  \to 0$.

Hence, using (\ref{eq: lin1U}) and (\ref{eq: lin2U}), we obtain
$$
\| \sqrt{n}\alpha(x)' ( \widehat \beta - \beta) -\alpha(x)' \Omega^{1/2}\mathcal{N}_k \|  \leq  \left\| \frac{1}{\sqrt{n}} \sum_{i=1}^n \alpha(x)'\Omega^{1/2}\zeta_i - \alpha(x)'\Omega^{1/2}\mathcal{N}_k \right\| +  \bar{R}_{1n}=o_P(a_n^{-1})
$$
uniformly over $x\in\mathcal{X}$. Since $\|\alpha(x)'\Omega^{1/2}\|$ is bounded from below uniformly over $x\in\mathcal{X}$, we conclude that (\ref{eq: strong approximation process 1}) holds, and (\ref{eq: strong approximation process 2}) is a direct consequence of (\ref{eq: strong approximation process 1}).

Further, under the assumption that $\sup_{x\in\mathcal{X}}n^{1/2}|r(x)|/\|s(x)\|=o_P(a_n^{-1})$,
$$
\frac{\sqrt{n}p(x)'(\widehat\beta-\beta)}{\|s(x)\|}-\frac{\sqrt{n}(\widehat{g}(x)-g(x))}{\|s(x)\|}=o_P(a_n^{-1}),
$$
so that (\ref{eq: strong approximation process 3}) follows. This completes the proof of the theorem.
\end{proof}

\begin{proof}[Proof of Theorem \ref{Thm:MainBootstrap}]
Note that $\widehat \beta^b$ solves the least squares  problem for the rescaled data $\{(\sqrt{h_i}y_i,\sqrt{h_i}p_i): i=1,\dots,n\}$. The weight $h_i$ is independent of $(y_i,p_i)$, $E[h_i]=1$, $E[h_i^2]=1$, $E[h_i^{m/2}]\lesssim 1$, and $\max_{1\leq i \leq n}h_i \lesssim_P \log n$. Thus, considering the model
$$
\sqrt{h}_i y_i=(\sqrt{h}_i p_i)'\beta +\sqrt{h_i} r_i +\sqrt{h}_i \epsilon_i
$$
allows us to extend all results from $\widehat \beta$ to $\widehat \beta^b$ replacing $\xi_k$ by $\xi_k^b = \xi_k(\log n)^{1/2}$ and $\ell_k c_k$ by $\ell_kc_k(\log n)^{1/2}$ and noting that $E[\max_{1\leq i\leq n}|\sqrt{h}_i\epsilon_i||X]\lesssim_P n^{1/m}(\log n)^{1/2}$. Also, since $\xi_k\geq k^{1/2}$, condition $\xi_k^{2m/(m-2)}\log k/n\lesssim 1$ assumed in A.5 implies that $\log k\lesssim \log n$.

Now, we apply Lemma \ref{lemma: uniform linearization} to the original problem (\ref{eq: original problem}) and to the weighted problem (\ref{eq: weighted bootstrap problem}). Then
\begin{align*}
\sqrt{n} \alpha(x)'\left(\widehat \beta^b -
\widehat \beta\right) &  = \sqrt{n}\alpha(x)' \left(\widehat \beta^b - \beta\right) + \sqrt{n} \alpha(x)'\left(\beta - \widehat \beta\right) \\
&  = \alpha(x)'\Gn[(h_i-1)p_i(\epsilon_i+r_i)] + R_{1n}^b(\alpha(x))
\end{align*}
where
$$
R_{1n}^b(\alpha(x)) \lesssim_P \sqrt{\frac{\xi_k^2\log^3n}{n}}(n^{1/m}\sqrt{\log n}+\sqrt{k}\ell_kc_k)
$$
uniformly over $x\in\mathcal{X}$, and so (\ref{eq: weighting bootstrap approximation 1}) follows.

Further, (\ref{eq: weighted bootstrap approximation 2}) follows similarly to Theorem \ref{theorem: uniform normality} by applying Yurinskii's
coupling for the weighted process with weights $v_i = h_i-1$ so that $E[v_i^2]=1$ and $E[|v_i|^3]\lesssim 1$. Thus there is a Gaussian random vector $\mathcal{N}_k\sim N(0,I_k)$  such that
\begin{equation}\label{eq: yurinski for bootstrap}
\left\| \frac{\Omega^{-1/2}}{\sqrt{n}}\sum_{i=1}^n(h_i-1)p_i(\epsilon_i+r_i) - \mathcal{N}_k   \right\| =o_P(a_n^{-1}).
\end{equation}
Combining (\ref{eq: yurinski for bootstrap}) with (\ref{eq: weighting bootstrap approximation 1}) yields (\ref{eq: weighted bootstrap approximation 2}) by the triangle inequality as in the proof of Theorem \ref{theorem: uniform normality}, and (\ref{eq: weighted bootstrap approximation 3}) follows from (\ref{eq: weighted bootstrap approximation 2}).

Note also that the results continue to hold in $P$-probability if we replace $P$ by $P^*(\cdot|D)$, since $B_n \lesssim_P 1$ implies that $B_n \lesssim_{P^*} 1$. Indeed, the first relation means that $P(|B_n| > \ell_n)= o(1)$ for any $\ell_n \to \infty$,
while the second means that $P^*(|B_n| > \ell_n)= o_{P}(1)$ for any $\ell_n \to \infty$. But the second clearly follows from the first by Markov inequality because $E[P^*(|B_n| > \ell_n)] =  P(|B_n| > \ell_n)= o(1)$.
\end{proof}

\begin{proof}[Proof of Theorem \ref{thm:Matrices}]
Note that it follows from $\bar{R}_{2n}\lesssim (\log k)^{1/2}$ that $\ell_kc_k\lesssim 1$ (see the definition of $\bar{R}_{2n}$ in (\ref{eq: lin4U})). Therefore, $\|\Sigma\|\lesssim (1+(\ell_kc_k)^2)\|Q\|\lesssim 1$. In addition, it follows from Condition A.4 that $v_n\lesssim n^{1/m}$, and so $\bar{R}_{1n}\lesssim (\log k)^{1/2}$ implies that
$$
(v_n \vee 1+\ell_kc_k)\sqrt{\frac{\xi_k^2\log k}{n}} \to 0.
$$
Further, the first result follows from the Markov inequality and Matrix LLN (Lemma \ref{rudelson}), which shows that $E[\|\widehat Q - Q\|] \lesssim (\xi_k^2\log k /n)^{1/2}\to 0$.

To establish the second result, we note that
 \begin{equation}\label{Eq:M01}\widehat \Sigma - \Sigma = \En[(\widehat\epsilon_i^2-\{\epsilon_i+r_i\}^2) p_ip_i' ] + \En[\{\epsilon_i+r_i\}^2 p_ip_i' ]-\Sigma.\end{equation}
The first term on the right hand side of (\ref{Eq:M01}) satisfies
\begin{align*}
&  \| \En [ (\widehat \epsilon_i^2 - \{\epsilon_i+r_i\}^2) p_ip_i']\|  \leq   \| \En[ \{ p_i'(\widehat \beta - \beta)\}^2 p_i p_i'] \|  + 2 \| \En[ (\epsilon_i+r_i) p_i'(\widehat \beta - \beta) p_i p_i'] \| \\ 
&  \qquad \leq   \max_{1\leq i\leq n} |p_i'(\widehat\beta - \beta)|^2 \| \En [p_ip_i'] \| +  \max_{1\leq i \leq n } (|\epsilon_i|+|r_i|) \max_{1\leq i\leq n} |p_i'(\widehat\beta - \beta)|  \|  \En [p_i p_i'] \| \\
& \qquad \lesssim_P \|\widehat Q\| \frac{\xi_k^2(\sqrt{\log k} + \bar{R}_{1n} + \bar{R}_{2n})^2}{n}+(v_n \vee 1+\ell_kc_k)\|\widehat Q\|\frac{\xi_k(\sqrt{\log k} + \bar{R}_{1n} + \bar{R}_{2n})}{\sqrt{n}}
\end{align*}
since  $\max_{1\leq i\leq p} |p_i'(\widehat\beta - \beta)|^2\lesssim_P \xi_k^2 (\sqrt{\log k} + \bar{R}_{1n}+\bar{R}_{2n})^2/n$ by Theorem \ref{theorem: uniform rate},  $\max_{1\leq i\leq n}|r_i| \leq \ell_kc_k$, and
$\max_{1\leq i \leq n } |\epsilon_i|^2\lesssim_P v_n^2$ by Markov's inequality. Therefore,
$$
\| \En [ (\widehat \epsilon_i^2 - \{\epsilon_i+r_i\}^2) p_ip_i']\|\lesssim_P (v_n \vee 1+\ell_kc_k)\sqrt{\frac{\xi_k^2\log k}{n}}
$$
because $\bar{R}_{1n}+\bar{R}_{2n}\lesssim (\log k)^{1/2}$, $\|\widehat{Q}\| \lesssim_{\mathrm{P}} 1$  by the first result, $\xi_k^2\log k/n\to 0$, and $v_n \vee 1+\ell_kc_k$ is bounded away from zero.

To control the second term in (\ref{Eq:M01}), let $\eta_1,\dots,\eta_n$ be a sequence of independent Rademacher random variables ($P(\eta_1=1)=P(\eta_1=-1)=1/2$) that are independent of the data. Then for $\eta=(\eta_1,\dots,\eta_n)$,
\begin{align*}
 &E\left[\|\En[\{\epsilon_i+r_i\}^2 p_ip_i' ]-\Sigma\|\right]\\
 & \qquad \lesssim E\left[E_\eta\left[ \|\En[ \eta_i\{\epsilon_i+r_i\}^2 p_ip_i' ]\|\right]\right]\\
& \qquad \lesssim \sqrt{\frac{\log k}{n}}E\left[ \left(\|\En[\{\epsilon_i+r_i\}^4\|p_i\|^2p_ip_i']\|\right)^{1/2}\right]\\
 & \qquad \leq \sqrt{\frac{\xi_k^2\log k}{n}}E\left[\max_{1\leq i\leq n}|\epsilon_i+r_i| \left(\|\En[\{\epsilon_i+r_i\}^2p_ip_i']\|\right)^{1/2}\right]\\
 & \qquad \leq \sqrt{\frac{\xi_k^2\log k}{n}} \left(E\left[\max_{1\leq i\leq n}|\epsilon_i+r_i|^2\right]\right)^{1/2}\left(E\left[\|\En[\{\epsilon_i+r_i\}^2p_ip_i']\|\right]\right)^{1/2}
\end{align*}
where the first inequality holds by Symmetrization Lemma (see Lemma 2.3.6 in \cite{vdV-W}), the second by Khinchin's inequality (Lemma \ref{lem: Khinchin}), the third by $\max_{1\leq i\leq n}\|p_i\|\leq \xi_k$, and the fourth by the Cauchy-Schwartz inequality.

Since for any positive numbers $a$, $b$, and $R$, $a \leq R(a+b)^{1/2}$ implies $a\leq R^2+R\sqrt{b}$, the expression above using the triangle inequality yields
$$
E\left[\|\En[\{\epsilon_i+r_i\}^2 p_i p_i' ]-\Sigma\|\right] \lesssim \frac{\xi_k^2\log k}{n}(v_n^2 + \ell_k^2c_k^2) + \left(\frac{\xi_k^2\log k}{n}\{v_n^2 + \ell_k^2c_k^2\}\right)^{1/2}\|\Sigma\|^{1/2},
$$
and so
$$
E\left[\|\En[\{\epsilon_i+r_i\}^2 p_ip_i' ]-\Sigma\|\right]\lesssim (v_n \vee 1+\ell_kc_k)\sqrt{\frac{\xi_k^2\log k}{n}}
$$
because $\|\Sigma\|\lesssim 1$ and $(v_n^2+\ell_k^2c_k^2)\xi_k^2\log k/n \to 0$.
Now, the second result follows from Markov's inequality.

Finally, we have
$$
\|\widehat{\Omega}-\Omega\|\lesssim \|(\widehat{Q}^{-1}-Q^{-1})\widehat{\Sigma}\widehat{Q}^{-1}\|+\|Q^{-1}(\widehat{\Sigma}-\Sigma)\widehat{Q}^{-1}\|+\|Q^{-1}\Sigma(\widehat{Q}^{-1}-Q^{-1})\|=o_P(1/a_n)
$$
whenever $\|\widehat{Q}-Q\|=o_P(1/a_n)$ and $\|\widehat{\Sigma}-\Sigma\|=o_P(1/a_n)$ because eigenvalues of both $Q$ and $\Sigma$ are bounded away from zero and from above. We can set $a_n=(v_n \vee 1+\ell_kc_k)(\xi_k^2\log k/n)^{1/2}$. This gives the third result of the theorem and completes the proof.
\end{proof}

\subsection{Proofs of Section \ref{Sec:LinearFunctionals}}
\begin{proof}[Proof of Lemma \ref{lem: variance consistency}]
As in the proof of Theorem \ref{theorem: pointwise}, all eigenvalues of $\Omega$ are bounded away from zero. Therefore,
\begin{equation}\label{eq: sigmas}
\left|\frac{\widehat{\sigma}_\theta(w)}{\sigma_\theta(w)}-1\right|
\leq \left|\frac{\widehat{\sigma}_\theta(w)^2}{\sigma_\theta(w)^2}-1\right|
= \frac{\|\ell_\theta(w)^\prime(\widehat{\Omega}-\Omega)\ell_\theta(w)\|}{\|\ell_\theta(w)^\prime\Omega \ell_\theta(w)\|}\lesssim_P\|\widehat{\Omega}-\Omega\|.
\end{equation}
In addition, by Theorem \ref{thm:Matrices},
\begin{equation}\label{eq: A.35}
\|\widehat{\Omega}-\Omega\|\lesssim_P (v_n \vee 1+\ell_kc_k)\sqrt{\frac{\xi_k^2\log k}{n}}=o(1).
\end{equation}
Combining (\ref{eq: sigmas}) and (\ref{eq: A.35}) gives the asserted claim.
\end{proof}

\subsection{Proofs of Section \ref{Sec:LinearPointwise}} 
\begin{proof}[Proof of Theorem \ref{theorem: pointwise rate}]
Fix $w\in\mathcal{I}$. Denote $\alpha:=\ell_\theta(w)/\|\ell_\theta(w)\|$. Then
\begin{align*}
| \widehat \theta(w) -  \theta(w)| &\leq | \ell_\theta(w)'(\widehat \beta -  \beta)| + |r_\theta(w)| \\
& \leq |\ell_\theta(w)'\mathbb{G}_n[p_i\epsilon_i]|/\sqrt{n} + \|\ell_\theta(w)\|\left(|R_{1n}(\alpha)|+|R_{2n}(\alpha)|\right)/\sqrt{n} + o(\|\ell_\theta(w)\|/\sqrt{n})
\end{align*}
where the second line follows from  Lemma \ref{lemma:linearization} and condition (i).
Next, note that by Lemma \ref{lemma:linearization},
$$
|R_{1n}(\alpha)|+|R_{2n}(\alpha)| \lesssim_P \sqrt{\frac{\xi_k^2\log k}{n}}(1+\sqrt{k}\ell_kc_k) + \ell_kc_k=o(1)
$$
where the last conclusion holds from conditions (iii) and (iv).
Finally, condition (ii) implies that
$$
E[|\ell_\theta(w)'\mathbb{G}_n[p_i\epsilon_i]|^2]\lesssim \|\ell_\theta(w)\|^2\bar{\sigma}^2\|Q\| \lesssim \|\ell_\theta(w)\|^2,
$$
and so the result follows by applying Chebyshev's inequality.
\end{proof}

\begin{proof}[Proof of Theorem \ref{Thm:DistributionsInferentialpointwise}] 

Under our conditions, all eigenvalues of $\Omega$ are bounded away from zero. Therefore,
$$
\frac{r_{\theta}(w)}{\widehat{\sigma}_{\theta}(w)}\lesssim_P \frac{r_{\theta}(w)}{\sigma_{\theta}(w)}\lesssim \frac{\sqrt{n}r_{\theta}(w)}{\|\ell_{\theta}(w)\|}\to 0
$$
where the first inequality follows from Lemma \ref{lem: variance consistency}. In addition, by Theorem \ref{theorem: pointwise},
$$
\frac{\ell_{\theta}(w)'(\widehat{\beta}-\beta)}{\sigma_{\theta}(w)}\to_d N(0,1).
$$
Hence,
$$
t(w)=\frac{\ell_{\theta}(w)'(\widehat{\beta}-\beta)}{\widehat{\sigma}_{\theta}(w)}-\frac{r_{\theta}(w)}{\widehat{\sigma}_{\theta}(w)}=\frac{\ell_{\theta}(w)'(\widehat{\beta}-\beta)}{(1+o_P(1))\sigma_{\theta}(w)}+o_P(1)\to_d N(0,1)
$$
by Slutsky's lemma. This completes the proof of the theorem.
\end{proof}

\subsection{Proofs of Section \ref{Sec:LinearUniform}} 
\begin{proof}[Proof of Lemma \ref{lem: uniform linearization functionals}]
By the triangle inequality,
\begin{align*}
\left\|\frac{\ell_{\theta}(w_1)}{\|\ell_{\theta}(w_1)\|}-\frac{\ell_{\theta}(w_2)}{\|\ell_{\theta}(w_2)\|}\right\|&\leq \frac{\|\ell_{\theta}(w_1)-\ell_{\theta}(w_2)\|}{\|\ell_{\theta}(w_1)\|}+\|\ell_{\theta}(w_2)\|\left|\frac{1}{\|\ell_{\theta}(w_1)\|}-\frac{1}{\|\ell_{\theta}(w_2)\|}\right|\\
&\leq \frac{2\|\ell_{\theta}(w_1)-\ell_{\theta}(w_2)\|}{\|\ell_{\theta}(w_1)\|}\lesssim \xi_{k,\theta}^L\|w_1-w_2\|
\end{align*}
uniformly over $w_1,w_2\in\mathcal{I}$ where the last inequality follows from the definition of $\xi_{k,\theta}^L$ and the condition that $1/\|\ell_{\theta}(w)\|\lesssim 1$ uniformly over $w\in\mathcal{I}$. Therefore, the proof follows from the same arguments as those given for Lemma \ref{lemma: uniform linearization}.
\end{proof}

\begin{proof}[Proof of Theorem \ref{theorem: Linearuniform rate}]
Given discussion in the proof of Lemma \ref{lem: uniform linearization functionals}, (\ref{eq: uniform linearization functional 1}) follows from the same arguments as those used for Theorem \ref{theorem: uniform rate}, equation (\ref{eq: 4.13}).

Now we prove (\ref{eq: uniform linearization functional 2}). By the triangle inequality,
\begin{equation}\label{eq: triangle inequality 9}
\sup_{w\in \I} | \widehat \theta(w) - \theta(w) |  \leq \sup_{w\in \I}| \ell_{\theta}(w)'(\widehat \beta - \beta)| + \sup_{w\in \I}|r_\theta(w)|.
\end{equation}
Further,
\begin{equation}\label{eq: residual 9}
\sup_{w\in\mathcal{I}}|r_\theta(w)|\leq \sup_{w\in\mathcal{I}}\frac{|r_n(w)|}{\|\ell_{\theta}(w)\|}\sup_{w\in\mathcal{I}}\|\ell_{\theta}(w)\|\lesssim \sqrt{\frac{\xi_{k,\theta}^2\log k}{n}}
\end{equation}
by the condition (ii) and the definition of $\xi_{k,\theta}$.
In addition, by Lemma \ref{lem: uniform linearization functionals} and (\ref{eq: uniform linearization functional 1}),
\begin{align}
\sup_{w\in\mathcal{I}} | \ell_{\theta}(w)'(\widehat \beta - \beta)| & \lesssim_P\frac{1}{\sqrt{n}}\left(\left|\sup_{w\in\I}\alpha_{\theta}(w)'\mathbb{G}_n[p_i\epsilon_i]\right|+\bar{R}_{1n}+\bar{R}_{2n}\right) \sup_{w\in\mathcal{I}}\|\ell_{\theta}(w)\|\label{Eq:AuxUnif1}\\
&\lesssim_P \sqrt{\frac{\log k}{n}}\sup_{w\in\mathcal{I}}\|\ell_{\theta}(w)\|\lesssim \sqrt{\frac{\xi_{k,\theta}^2\log k}{n}}.\label{Eq:AuxUnif2}
\end{align}
Combining (\ref{eq: triangle inequality 9}), (\ref{eq: residual 9}), (\ref{Eq:AuxUnif1}), and (\ref{Eq:AuxUnif2}) gives the asserted claim.
\end{proof}

\begin{proof}[Proof of Theorem \ref{Thm:DistributionsInferential}]
Since $\bar{R}_{1n}=o_P(a_n^{-1})$, we have
$$
(v_n \vee 1+\ell_kc_k)\frac{\xi_k\log k}{\sqrt{n}}=o(a_n^{-1}).
$$
Further, as in the proof of Theorem \ref{theorem: uniform normality} and using Lemma \ref{lem: uniform linearization functionals}, we can find $\mathcal{N}_k\sim N(0,I_k)$ such that
$$
\left\|\sqrt{n}\alpha_{\theta}(w)'(\widehat{\beta}-\beta)-\alpha_{\theta}(w)'\Omega^{1/2}\mathcal{N}_k\right\|=o_P(a_n^{-1})
$$
uniformly over $w\in\I$. Since $\|\alpha_{\theta}(w)'\Omega^{1/2}\|$ is bounded away from zero uniformly over $w\in\I$,
$$
\left\|\sqrt{n}\frac{\ell_{\theta}(w)'(\widehat{\beta}-\beta)}{\|\ell_{\theta}(w)'\Omega^{1/2}\|}-\frac{\ell_{\theta}(w)'\Omega^{1/2}\mathcal{N}_k}{\|\ell_{\theta}(w)'\Omega^{1/2}\|}\right\|=o_P(a_n^{-1}),
$$
or, equivalently,
$$
\left\|\frac{\ell_{\theta}(w)'(\widehat{\beta}-\beta)}{\sigma_{\theta}(w)}-\frac{\ell_{\theta}(w)'\Omega^{1/2}\mathcal{N}_k/\sqrt{n}}{\sigma_{\theta}(w)}\right\|=o_P(a_n^{-1}),
$$
uniformly over $w\in\I$. Further,
\begin{align*}
\left|\frac{\ell_{\theta}(w)'(\widehat{\beta}-\beta)}{\sigma_{\theta}(w)}-\frac{\ell_{\theta}(w)'(\widehat{\beta}-\beta)}{\widehat{\sigma}_{\theta}(w)}\right| & \leq \frac{|\ell_{\theta}(w)'(\widehat{\beta}-\beta)|}{\sigma_{\theta}(w)}\left|1-\frac{\sigma_{\theta}(w)}{\widehat{\sigma}_{\theta}(w)}\right|\\
& \lesssim \sqrt{n}|\alpha_{\theta}(w)'(\widehat{\beta}-\beta)|\left|1-\frac{\sigma_{\theta}(w)}{\widehat{\sigma}_{\theta}(w)}\right|\\
& \lesssim_P \sqrt{\log k}(v_n \vee 1+\ell_kc_k)\sqrt{\frac{\xi_k^2\log k}{n}}=o(a_n^{-1})
\end{align*}
uniformly over $w\in\I$ where the second line follows from $\|\alpha_{\theta}(w)'\Omega^{1/2}\|$ being bounded away from zero uniformly over $w\in \I$ and the third line follows from Lemmas \ref{lem: variance consistency} and \ref{lem: uniform linearization functionals} and Theorem \ref{theorem: Linearuniform rate}. Therefore,
\begin{equation}\label{eq: strong approximation close}
\left\|\frac{\ell_{\theta}(w)'(\widehat{\beta}-\beta)}{\widehat\sigma_{\theta}(w)}-\frac{\ell_{\theta}(w)'\Omega^{1/2}\mathcal{N}_k/\sqrt{n}}{\sigma_{\theta}(w)}\right\|=o_P(a_n^{-1})
\end{equation}
uniformly over $w\in\I$. In addition, $\sup_{w\in\I}|r_{\theta}(w)|/\sigma_{\theta}(w)=o_P(a_n^{-1})$ uniformly over $w\in\I$ and Lemma \ref{lem: variance consistency} imply that $\sup_{w\in\I}|r_{\theta}(w)|/\widehat{\sigma}_{\theta}(w)=o_P(a_n^{-1})$, and so it follows from (\ref{eq: strong approximation close}) that
$$
\left\|\frac{\widehat{g}(w)-g(w)}{\widehat\sigma_{\theta}(w)}-\frac{\ell_{\theta}(w)'\Omega^{1/2}\mathcal{N}_k/\sqrt{n}}{\sigma_{\theta}(w)}\right\|=o_P(a_n^{-1})
$$
uniformly over $w\in\I$. This completes the proof of the theorem.
\end{proof}

\begin{proof}[Proof of Theorem \ref{thm: strong approximation suprema}]
We have
\begin{equation}\label{eq: triangle inequality 11}
\frac{\widehat{\theta}(w)-\theta(w)}{\widehat{\sigma}_\theta(w)}=\frac{\ell_{\theta}(w)^\prime(\widehat{\beta}-\beta)}{\widehat{\sigma}_\theta(w)}-\frac{r_\theta(w)}{\widehat{\sigma}_\theta(w)}.
\end{equation}
Under the condition $\bar{R}_{1n}+\bar{R}_{2n}\lesssim 1/(\log k)^{1/2}$,
\begin{equation}\label{eq: linearization 11}
\left|\frac{\ell_{\theta}(w)^\prime(\widehat{\beta}-\beta)}{\widehat{\sigma}_n(w)}-\frac{\ell_{\theta}(w)^\prime(\widehat{\beta}-\beta)}{\sigma_{\theta}(w)}\right| = o_P \left(\frac{1}{\sqrt{\log k}}\right)
\end{equation}
uniformly over $w\in\I$ by the argument used in the proof of Theorem \ref{Thm:DistributionsInferential} with $a_n=1/(\log k)^{1/2}$.
Further, by Lemma \ref{lem: uniform linearization functionals},
\begin{equation}\label{eq: linearization 11-3}
\frac{\ell_{\theta}(w)^\prime(\widehat{\beta}-\beta)}{\sigma_{\theta}(w)}=\frac{\ell_{\theta}(w)^\prime\mathbb{G}_n[p_i\epsilon_i]}{\sqrt{n}\sigma_{\theta}(w)}+o_P\left(\frac{1}{\sqrt{\log k}}\right)
\end{equation}
uniformly over $w\in\mathcal{I}$ since $\bar{R}_{1n}+\bar{R}_{2n}\lesssim 1/(\log k)^{1/2}$.
In addition, as in the proof of Theorem \ref{Thm:DistributionsInferential} with $a_n=1/(\log k)^{1/2}$,
\begin{equation}\label{eq: linearization 11-4}
\frac{|r_{\theta}(w)|}{\widehat{\sigma}_{\theta}(w)}= o_P\left( \frac{1}{\sqrt{\log k}}\right)
\end{equation}
uniformly over $w\in\I$. Combining (\ref{eq: triangle inequality 11}), (\ref{eq: linearization 11}), (\ref{eq: linearization 11-3}), and (\ref{eq: linearization 11-4}) yields
\begin{equation}\label{eq: final linearization}
\frac{\widehat{\theta}(w)-\theta(w)}{\widehat{\sigma}_\theta(w)}=\frac{\ell_{\theta}(w)^\prime\mathbb{G}_n[p_i\epsilon_i]}{\sqrt{n}\sigma_{\theta}(w)}+o_P\left(\frac{1}{\sqrt{\log k}}\right).
\end{equation}
Now, under the condition $\xi_k\log^2 k/n^{1/2-1/m} \to 0$, the asserted claim follows from Proposition 3.3 in \cite{CCK2012} applied to the first term on the right hand side of (\ref{eq: final linearization}) (note that Proposition 3.3 in \cite{CCK2012} only considers a special case where $\ell_{\theta}(w)$, $w\in\mathcal{I}$, is replaced by $p(x)$, $x\in\mathcal{X}$, but the same proof applies for a more general case studied here, with $\ell_{\theta}(w)$, $w\in\mathcal{I}$).
\end{proof}

\begin{proof}[Proof of Theorem \ref{thm: confidence bands validity}]
The proof consists of two steps. The asserted claims are proven in Step 1, and Step 2 contains some intermediate calculations.

\textbf{Step 1.}
Under our conditions, it follows from Step 2 that there exists a sequence $\{\varepsilon_n\}$ such that $\varepsilon_n=o(1)$ and
\begin{equation}\label{eq: A.43}
P\left\{\left|\sup_{w\in\I}|\widehat{t}_n^*(w)|-\sup_{w\in\I}|t_n^*(w)|\right|>\varepsilon_n/\sqrt{\log k}\right\}=o(1).
\end{equation}
Let $c_n^0(1-\alpha)$ denote the $(1-\alpha)$-quantile of $\sup_{w\in\I}|t_n^*(w)|$.  Then in view of (\ref{eq: A.43}), Lemma \ref{lemma: quantiles are close} implies that there exists a sequence $\{\nu_n\}$ such that $\nu_n=o(1)$ and
\begin{align}
&P\left\{c_n(1-\alpha)<c_n^0(1-\alpha-\nu_n)-\varepsilon_n/\sqrt{\log k}\right\}=o(1),\label{eq: A.44a}\\
&P\left\{c_n(1-\alpha)>c_n^0(1-\alpha+\nu_n)+\varepsilon_n/\sqrt{\log k}\right\}=o(1)\label{eq: A.44b}.
\end{align}
Further, it follows from Theorem \ref{thm: strong approximation suprema} that there exists a sequence $\{\beta_n\}$ of constants and a sequence $\{Z_n\}$ of random variables such that $\beta_n=o(1)$, $Z_n$ equals in distribution to $\|t_n^*\|_{\I}$, and
\begin{equation}\label{eq: A.45}
P\left\{\left|\sup_{w\in\I}|t_n(w)|-Z_n\right|>\beta_n/\sqrt{\log k}\right\}=o(1).
\end{equation}
Hence, for some universal constant $A$,
\begin{align*}
P(\sup_{w\in\I}|t_n(w)|\leq c_n(1-\alpha))&\leq P(Z_n\leq c_n(1-\alpha)+\beta_n/\sqrt{\log k})+o(1)\\
&\leq P(Z_n\leq c_n^0(1-\alpha+\nu_n)+(\varepsilon_n+\beta_n)/\sqrt{\log k})+o(1)\\
&\leq P(Z_n\leq c_n^0(1-\alpha+\nu_n+A(\varepsilon_n+\beta_n)))+o(1)\\
&=1-\alpha+\nu_n+A(\varepsilon_n+\beta_n)+o(1)\nonumber\\
&=1-\alpha+o(1)\nonumber
\end{align*}
where the first inequality follows from (\ref{eq: A.45}), the second from (\ref{eq: A.44b}), and the third from Lemma \ref{lem: anticoncentration}. This gives one side of the bound in (\ref{eq: supremum validity}). The other side of the bound can be proven by a similar argument. Therefore, (\ref{eq: supremum validity}) follows. Further, (\ref{eq: confidence bands validity}) is a direct consequence of (\ref{eq: supremum validity}).

Finally, we consider (\ref{eq: confidence bands width}). The second inequality in (\ref{eq: confidence bands width}) holds because $\sigma_{\theta}(w)\lesssim \|\ell_{\theta}(w)\|/n^{1/2}$ since all eigenvalues of $\Omega$ are bounded from above. To prove the first inequality, note that by Lemma \ref{lem: variance consistency}, $\widehat{\sigma}_\theta(w)/\sigma_{\theta}(w)=1+o_P(1)$ uniformly over $w\in\mathcal{I}$. In addition, Step 2 shows that
\begin{equation}\label{eq: A.50}
c_n(1-\alpha)\lesssim_P \sqrt{\log k}.
\end{equation}
Therefore,
$
2c_n(1-\alpha)\widehat{\sigma}_n(w)\lesssim_P (\log k)^{1/2}\sigma_{\theta}(w),
$
uniformly over $w\in\mathcal{I}$,
which is the first inequality in (\ref{eq: confidence bands width}). To complete the proof, we provide auxilliary calculations in Step 2.

\textbf{Step 2.} We first prove (\ref{eq: A.43}). Note that
$$
\left|\sup_{w\in\I}|\widehat{t}_n^*(w)|-\sup_{w\in\I}|t_n^*(w)|\right|\leq \sup_{w\in\mathcal{I}}\left|\widehat{t}_n^{*}(w)-t_n^{*}(w)\right|=\sup_{w\in\mathcal{I}}\left|\left(\frac{\ell_{\theta}(w)^\prime\widehat{\Omega}^{1/2}}{\sqrt{n}\widehat{\sigma}_n(w)}-\frac{\ell_{\theta}(w)^\prime\Omega^{1/2}}{\sqrt{n}\sigma_{\theta}(w)}\right)\mathcal{N}_k\right|.
$$
Denote $T_n(w):=\widehat{t}_n^{*}(w)-t_n^{*}(w)$. Then, conditional on the data,
$
\{T_n(w),\, w\in\mathcal{I}\}
$
is a zero-mean Gaussian process. Further, we have for $E_{\mathcal{N}_k}[\cdot]$ denoting the expectation with respect to the distribution of $\mathcal{N}_k$,
\begin{align*}
E_{\mathcal{N}_k}[T_n(w)^2]^{1/2}&=\left\| \frac{\ell_{\theta}(w)^\prime\widehat{\Omega}^{1/2}}{\sqrt{n}\widehat{\sigma}_n(w)}-\frac{\ell_{\theta}(w)^\prime\Omega^{1/2}}{\sqrt{n}\sigma_{\theta}(w)}\right\|\\
&\leq \frac{\|\ell_{\theta}(w)\|}{\sqrt{n}\widehat{\sigma}_n(w)}\|\widehat{\Omega}^{1/2}-\Omega^{1/2}\|+\left\|\frac{\ell_{\theta}(w)^\prime\Omega^{1/2}}{\sqrt{n}\sigma_{\theta}(w)}\right\|\left|\frac{\sigma_{\theta}(w)}{\widehat{\sigma}_n(w)}-1\right|\\
&\lesssim_P \|\widehat{\Omega}^{1/2}-\Omega^{1/2}\|+\left|\frac{\sigma_{\theta}(w)}{\widehat{\sigma}_n(w)}-1\right|\\
&\lesssim_P \|\widehat{\Omega}-\Omega\|=o_P\left(\frac{1}{\sqrt{\log k}}\right)
\end{align*}
uniformly over $w\in\mathcal{I}$ where the last line follows from Lemma \ref{lem: matrices}. In addition, uniformly over $w_1,w_2\in\mathcal{I}$,
\begin{align*}
& E_{\mathcal{N}_k}[(T_n(w_1)-T_n(w_2))^2]^{1/2}\leq \\
 &\qquad \leq \left\| \frac{\ell_{\theta}(w_1)^\prime\widehat{\Omega}^{1/2}}{\sqrt{n}\widehat{\sigma}_n(w_1)}-\frac{\ell_{\theta}(w_2)^\prime\widehat{\Omega}^{1/2}}{\sqrt{n}\widehat{\sigma}_n(w_2)}\right\|+\left\| \frac{\ell_{\theta}(w_1)^\prime\Omega^{1/2}}{\sqrt{n}\sigma_{\theta}(w_1)}-\frac{\ell_{\theta}(w_2)^\prime\Omega^{1/2}}{\sqrt{n}\sigma_{\theta}(w_2)}\right\|\\
&\qquad \lesssim_P \left\| \frac{\ell_{\theta}(w_1)}{\sqrt{n}\widehat{\sigma}_n(w_1)}-\frac{\ell_{\theta}(w_2)}{\sqrt{n}\widehat{\sigma}_n(w_2)}\right\|+\left\| \frac{\ell_{\theta}(w_1)}{\sqrt{n}\sigma_{\theta}(w_1)}-\frac{\ell_{\theta}(w_2)}{\sqrt{n}\sigma_{\theta}(w_2)}\right\|.
\end{align*}
Moreover, uniformly over $w_1,w_2\in\mathcal{I}$,
\begin{align*}
\left\| \frac{\ell_{\theta}(w_1)}{\sqrt{n}\widehat{\sigma}_n(w_1)}-\frac{\ell_{\theta}(w_2)}{\sqrt{n}\widehat{\sigma}_n(w_2)}\right\|&\leq \frac{\|\ell_{\theta}(w_1)-\ell_{\theta}(w_2)\|}{\sqrt{n}\widehat{\sigma}_n(w_1)}+\frac{\|\ell_{\theta}(w_2)\|}{\sqrt{n}}\left|\frac{1}{\widehat{\sigma}_n(w_1)}-\frac{1}{\widehat{\sigma}_n(w_2)}\right|\\
&= \frac{\|\ell_{\theta}(w_1)-\ell_{\theta}(w_2)\|}{\sqrt{n}\widehat{\sigma}_n(w_1)}+\frac{\|\ell_{\theta}(w_2)\|}{\sqrt{n}}\frac{|\widehat{\sigma}_n(w_2)-\widehat{\sigma}_n(w_1)|}{\widehat{\sigma}_n(w_1)\widehat{\sigma}_n(w_2)}\\
&\lesssim_P \frac{\|\ell_{\theta}(w_1)-\ell_{\theta}(w_2)\|}{\|\ell_{\theta}(w_1)\|}\lesssim \xi^L_{k,\theta}\|w_1-w_2\|
\end{align*}
where the last inequality follows from Condition A.6. A similar argument shows that
$$
\left\| \frac{\ell_{\theta}(w_1)}{\sqrt{n}\sigma_{\theta}(w_1)}-\frac{\ell_{\theta}(w_2)}{\sqrt{n}\sigma_{\theta}(w_2)}\right\|\lesssim_P\xi^L_{k,\theta}\|w_1-w_2\|
$$
uniformly over $w_1,w_2\in\mathcal{I}$. Now, (\ref{eq: A.43}) follows from Dudley's inequality (\cite{Dudley1967}).

Finally, to show (\ref{eq: A.50}), we note that in view of (\ref{eq: A.43}), it suffices to prove that
\begin{equation}\label{eq: A.51}
c_n^0(1-\alpha)\lesssim \sqrt{\log k}.
\end{equation}
But $\{t_n^{*}(w),\, w\in\mathcal{I}\}$ is a zero mean Gaussian process satisfying $E[t_n^{*}(w)^2]^{1/2}=1$ for all $w\in\mathcal{I}$ and
$$
E[(t_n^{*}(w_1)-t_n^{*}(w_2))^2]^{1/2}\leq \left\| \frac{\ell_{\theta}(w_1)^\prime\Omega^{1/2}}{\sqrt{n}\sigma_{\theta}(w_1)}-\frac{\ell_{\theta}(w_2)^\prime\Omega^{1/2}}{\sqrt{n}\sigma_{\theta}(w_2)}\right\|\lesssim \xi^L_{k,\theta}\|w_1-w_2\| 
$$
where the last inequality was shown above. Hence, (\ref{eq: A.51}) follows from combining Dudley's and Markov's inequalities.
\end{proof}

\subsection{Proofs of Section \ref{Sec:Tools}}


\begin{proof}[Proof of Lemma \ref{lem: Khinchin}]
The first part of the lemma, inequality (\ref{eq: khinchin main}), is proven in Section 3 of \cite{Rudelson1999}. To prove the second part of the lemma, inequality (\ref{eq: khinchin corollary}), observe that for $2\leq k\leq e^2$, the result is trivial. On the other hand, for $k>e^2$, we have
\begin{align*}\label{eq: khinchin corollary}
E_{\varepsilon}\left[\| \Gn[\varepsilon_i Q_i]\|\right] &\leq E_{\varepsilon}\left[\| \Gn[\varepsilon_i Q_i] \|_{S_{\log k}}\right] \leq \left(E_{\varepsilon}\left[\| \Gn[\varepsilon_i Q_i] \|_{S_{\log k}} ^{\log k}\right]\right)^{1/\log k}\\
&\lesssim \sqrt{\log k}  \left\| ( \En[ Q_i^2])^{1/2} \right\|_{S_{\log k}} \lesssim \sqrt{\log k} \left \| ( \En[ Q_i^2])^{1/2} \right\|
\end{align*}
where the first inequality follows from (\ref{equiv}), the second from Jensen's inequality, the third from the first part of the lemma, and the fourth from (\ref{equiv}) again. Related derivation can be also found in Section 3 of \cite{Rudelson1999}. This completes the proof of the lemma.
\end{proof}

\begin{proof}[Proof of Lemma \ref{rudelson}] Using the Symmetrization Lemma 2.3.6 in \cite{vdV-W} and the Khinchin inequality (Lemma \ref{lem: Khinchin}), bound
$$
\Delta  := E\left[\| \widehat Q -  Q  \|\right]\leq 2 E E_{\varepsilon} \left[\| \En [\varepsilon_i Q_i] \|\right] \lesssim \sqrt{\frac{\log k}{n}} E\left[ \| (\En Q_i^2 )^{1/2} \|\right].
$$
Also, observe that for any $\alpha\in S^{k-1}$,
$$
\alpha' Q_i^2\alpha\leq M\alpha' Q_i\alpha,
$$
so that
$$
\|\En [Q_i^2]\|\leq M\|\En [Q_i]\|.
$$
Therefore,
$$
E \left[\| (\En Q_i^2 )^{1/2} \|\right] = E \left[\| (\En Q_i^2 ) \|^{1/2}\right] \leq E \Big[(M \|\En Q_i\|)^{1/2} \Big]\leq  \Big [ M E \| \En Q_i \| \Big ]^{1/2}
$$
where the last assertion follows from Jensen's inequality. In addition, by the triangle inequality,
$$
E[\|\En Q_i\|]  \leq   \Delta +  \| Q \|.
$$
Hence,
$$
\Delta \lesssim \sqrt{ \frac{M \log k } {n} } [\Delta + \|Q\|]^{1/2}.
$$
Denoting $a:=M\log k/n$ and solving this inequality for $\Delta$ gives
$$
\Delta\lesssim a+\sqrt{a^2+a\|Q\|}\lesssim a+\sqrt{a\|Q\|}.
$$
This completes the proof of the lemma.
\end{proof}

\begin{proof}[Proof of Proposition \ref{prop1}]
For a  $\tau > 0$ specified later, define
$\epsilon_{i}^{-} := \epsilon_{i} I(| \epsilon_{i} | \leq \tau) - E [\epsilon_{i} I(| \epsilon_{i} | \leq \tau) | X_{i}]$ and
$\epsilon_{i}^{+} := \epsilon_{i} I(| \epsilon_{i} | > \tau) - E [\epsilon_{i} I(| \epsilon_{i} | > \tau) | X_{i}]$. Since $E[ \epsilon_{i} | X_{i} ] = 0$, $\epsilon_{i} = \epsilon_{i}^{-} + \epsilon_{i}^{+}$. Invoke the decomposition
\begin{equation*}
\sum_{i=1}^{n} \epsilon_{i} f(X_{i}) = \sum_{i=1}^{n} \epsilon^{-}_{i} f(X_{i}) + \sum_{i=1}^{n} \epsilon^{+}_{i} f(X_{i}).
\end{equation*}
We apply Theorem \ref{GK06} to the first term. Noting that $\var ( \epsilon_{i}^{-} f(X_{i}) ) \leq \sup_{x} E [ (\epsilon_{i}^{-})^{2} | X_{i}=x] E [ f(X_{i})^{2} ] \leq \sup_{x} E [ \epsilon_{i}^{2} | X_{i}=x] = \sigma^{2}$ and $| \epsilon_{i}^{-} f(X_{i}) | \leq 2 \tau b$, we have
\begin{equation*}
E \left [ \left \|  \sum_{i=1}^{n} \epsilon^{-}_{i} f(X_{i}) \right \|_{\mathcal{F}} \right] \leq C \left [ \sqrt{n \sigma^{2} V \log (Ab)} + V \tau b \log (Ab) \right ].
\end{equation*}
On the other hand, applying Theorem 2.14.1 of \cite{vdV-W} to the second term, we obtain
\begin{equation}
E \left [ \left \|  \sum_{i=1}^{n} \epsilon^{+}_{i} f(X_{i}) \right \|_{\mathcal{F}} \right] \leq C \sqrt{n} b \sqrt{E[ | \epsilon_{1}^{+} |^{2}]} \int_{0}^{1} \sqrt{V \log (A/\varepsilon)} d \varepsilon. \label{moment}
\end{equation}
By assumption,
\begin{equation*}
E[ | \epsilon_{1}^{+} |^{2}] \leq E[ \epsilon_{1}^{2} I( | \epsilon_{1} | > \tau ) ] \leq \tau^{-m+2} E[ | \epsilon_{1} |^{m} ],
\end{equation*}
by which we have
\begin{equation*}
(\ref{moment}) \leq C \sqrt{ E[ | \epsilon_{1} |^{m} ] } b \tau^{-m/2+1} \sqrt{nV \log (A)}.
\end{equation*}
Taking $\tau = b^{2/(m-2)}$, we obtain the desired inequality.
\end{proof}

\subsection{Additional technical results}



\begin{lemma}[Closeness in Probability Implies Closeness of Conditional Quantiles]\label{lemma: quantiles are close}
Let $X_n$ and $Y_n$ be random variables and $\mathcal{D}_n$ be a random vector. Let $F_{X_n}(x|\mathcal{D}_n)$ and $F_{Y_n}(x|\mathcal{D}_n)$
denote the conditional distribution functions, and  $F^{-1}_{X_n}(p|\mathcal{D}_n)$ and $F^{-1}_{Y_n}(p|\mathcal{D}_n)$
denote the corresponding conditional quantile functions. If $|X_n - Y_n| = o_P(\varepsilon)$,
 then for some $\nu_n \searrow 0$ with probability converging to one
$$
F^{-1}_{X_n}(p|\mathcal{D}_n) \leq F^{-1}_{Y_n}(p+\nu_n|\mathcal{D}_n) + \varepsilon \text{ and } F^{-1}_{Y_n}(p|\mathcal{D}_n) \leq F^{-1}_{X_n}(p+\nu_n|\mathcal{D}_n) + \varepsilon, \forall p \in (\nu_n, 1- \nu_n).
$$
\end{lemma}
\begin{proof}[Proof of Lemma \ref{lemma: quantiles are close}]  We have that for some $\nu_n \searrow 0$, $P\{ |X_n - Y_n| >\varepsilon\} = o(\nu_n)$.  This implies that  $P[P\{ |X_n - Y_n| >\varepsilon|\mathcal{D}_n\} \leq \nu_n] \to 1$, i.e.
there is a set $\Omega_n$  such that $P(\Omega_n) \to 1$
and $P\{ |X_n - Y_n| >\varepsilon|\mathcal{D}_n\} \leq \nu_n$ for all $\mathcal{D}_n \in \Omega_n$. So, for all $\mathcal{D}_n \in \Omega_n$
$$
F_{X_n}(x|\mathcal{D}_n) \geq F_{Y_n+ \varepsilon}(x|\mathcal{D}_n) - \nu_n \text { and } F _{Y_n}(x|\mathcal{D}_n) \geq F_{X_n+ \varepsilon}(x|\mathcal{D}_n) - \nu_n, \forall x \in \Bbb{R},
$$
which implies the inequality stated in the lemma, by definition of the conditional quantile function and equivariance of quantiles to location shifts. \end{proof}

\begin{lemma}\label{lem: matrices}
Let $A$ and $B$ be $k\times k$ symmetric positive semidefinite matrices. Assume that $B$ is positive definite. Then $\|A^{1/2}-B^{1/2}\|\leq \|A-B\|\|B^{-1}\|^{1/2}$.
\end{lemma}
\begin{proof}[Proof of Lemma \ref{lem: matrices}]
This is exercise 7.2.18 in \cite{HJ1990}. For completeness, we derive this result here. Let $a$ be an eigenvector of $E=A^{1/2}-B^{1/2}$ with eigenvalue $\lambda=\|A^{1/2}-B^{1/2}\|$.
Then
\begin{align*}
\|A-B\|&\geq |a^\prime(A-B)a|\\
&=|a^\prime(A^{1/2}E+EA^{1/2}-E^2)a|\\
&=|\lambda a^\prime(A^{1/2}+A^{1/2}-E)a|\\
&=\lambda|a^\prime(A^{1/2}+B^{1/2})a|\\
&\geq \lambda|\lambda_{\min}(A^{1/2})+\lambda_{\min}(B^{1/2})|
\end{align*}
where $\lambda_{\min}(P)$ denotes the minimal eigenvalue of $P$ for $P=A^{1/2}$ or $B^{1/2}$. Since $A$ is positive semidefinite, $\lambda_{\min}(A^{1/2})\geq 0$. Since $B$ is positive definite, $\lambda_{\min}(B^{1/2})=\|B^{-1}\|^{-1/2}$. Combining these bounds gives the asserted claim.
\end{proof}

\end{document}